\documentclass[lettersize,journal]{IEEEtran}

\IEEEoverridecommandlockouts

\usepackage{cite}
\usepackage{array}
\usepackage{soul}
\usepackage[table]{xcolor}
\usepackage{textcomp}
\usepackage{xcolor}
\usepackage{pifont}
\usepackage{setspace}
\usepackage[normalem]{ulem}
\usepackage{hyperref}
\hypersetup{
    colorlinks,
    linkcolor={black},
    citecolor={black},
    urlcolor={blue!80!black}
}
\usepackage{url}

\def\BibTeX{{\rm B\kern-.05em{\sc i\kern-.025em b}\kern-.08em
    T\kern-.1667em\lower.7ex\hbox{E}\kern-.125emX}}
\usepackage{balance}

\usepackage[pdftex]{graphicx}

\usepackage[font=small,compatibility=false]{caption}
\usepackage[font=footnotesize]{subcaption}

\usepackage{tabularx}
\usepackage{booktabs}

\usepackage{amsfonts,amsmath,amssymb,amsxtra,amsbsy,mathtools,cuted,bbold,amsthm}
\usepackage[cmintegrals]{newtxmath}
\usepackage{bm}
\usepackage{upgreek}
\usepackage{enumitem}

\newtheorem{corollary}{Corollary}

\newtheorem{remark}{Remark}

\newtheorem{assumption}{Assumption}
\newtheorem{example}{Example}
\newtheorem{definition}{Definition}

\newcommand{\E}[1]{\mathbb{E}\left\{ #1 \right\}} % expectation
\DeclareMathOperator*{\argmax}{arg\,max}    % argmax
\DeclareMathOperator{\diag}{diag}

\newcommand{\abs}[1]{\left\lvert#1\right\rvert} % absolute value
\newcommand{\ltwonorm}[1]{\left\lVert#1\right\rVert_2} % Euclidean norm
\newcommand{\chest}{\mathrm{chest}} % channel estimation
\newcommand{\comm}{\mathrm{comm}} % communication
\newcommand{\prob}{\mathrm{prob}} % probing
\newcommand{\refl}{\mathrm{refl}} % reflection

\usepackage{tikz}
\usetikzlibrary{patterns,shapes.arrows}

\usepackage{pgfplots}
\usepgfplotslibrary{groupplots,dateplot}
\pgfplotsset{
    compat=newest,
    legend style={
        font=\scriptsize,
        fill opacity=0.5,
        draw opacity=1,
        text opacity=1,
        draw=white!15!black,
        legend cell align=left,
        align=left
    }, 
    width=6.5cm,     
    yminorticks=false,
    xminorticks=false,
    title style={
        font=\small
    },
    label style={
        font=\footnotesize
    },
    tick style={color=black},
    tick label style={
        font=\scriptsize
    },
    grid style={
        dotted,
        draw=gray%!20
    },
    major grid style={
        dotted,
        draw=gray%!20
    }
}

\DeclareUnicodeCharacter{2212}{−}

\usepackage[acronym,nohypertypes={acronym,notation}]{glossaries}
\newacronym{5g}{5G}{fifth generation}
\newacronym{3d}{3D}{three dimensional}
\newacronym{aoa}{AoA}{angle of arrival}
\newacronym{awgn}{AWGN}{additive white Gaussian noise}
\newacronym{aod}{AoD}{angle of departure}
\newacronym{ap}{AP}{access point}
\newacronym{b5g}{B5G}{Beyond-5G}
\newacronym[plural=BSs, firstplural=base stations (BSs)]{bs}{BS}{base station}
\newacronym{cc}{CC}{control channel}
\newacronym{comm}{COMM}{communication}
\newacronym{chest}{CHEST}{channel estimation}
\newacronym{csi}{CSI}{channel state information}
\newacronym{cdf}{cdf}{cumulative distribution function}
\newacronym{dc}{DC}{direct current}
\newacronym{dsp}{DSP}{digital signal processing}
\newacronym{dl}{DL}{downlink}
\newacronym{doa}{DoA}{direction-of-arrival}
\newacronym{emf}{EMF}{electromagnetic field}
\newacronym{em}{EM}{electromagnetic}
\newacronym{fp}{FP}{fractional program}
\newacronym{glrt}{GLRT}{generalized likelihood ratio test}
\newacronym{hris}{HRIS}{hybrid \gls{ris}}
\newacronym{iid}{i.i.d.}{independent and identically distributed}
\newacronym{ios}{IoS}{Internet-of-Surfaces}
\newacronym{iot}{IoT}{Internet-of-Things}
\newacronym[plural=KPIs, firstplural=key performance indicators (KPIs)]{kpi}{KPI}{key performance indicator}
\newacronym{ls}{LS}{least-squares}
\newacronym{lf}{LF}{low frequency}
\newacronym{los}{LoS}{line-of-sight}
\newacronym{mac}{MAC}{medium access control}
\newacronym{mimo}{MIMO}{multiple-input multiple-output}
\newacronym{mmimo}{mMIMO}{massive multiple-input multiple-output}
\newacronym{miso}{MISO}{multiple-input single-output}
\newacronym{ml}{ML}{machine learning}
\newacronym{mle}{MLE}{maximum-likelihood estimator}
\newacronym{mmse}{MMSE}{minimum mean squared error}
\newacronym{mrc}{MRC}{maximum-ratio combining}
\newacronym{mse}{MSE}{mean-squared error}
\newacronym{nmse}{NMSE}{normalized mean-squared error}
\newacronym{nlos}{NLoS}{non-line-of-sight}
\newacronym{ofdm}{OFDM}{orthogonal frequency division multiplexing}
\newacronym{pdf}{PDF}{probability distribution function}
\newacronym{pla}{PLA}{planar linear array}
\newacronym{pap}{P\&P}{plug-and-play}
\newacronym{ppp}{PPP}{Poisson point process}
\newacronym{phy}{PHY}{physical}
\newacronym{ris}{RIS}{reconfigurable intelligent surface}

\newacronym[firstplural=radio frequencies (RFs)]{rf}{RF}{radio-frequency}
\newacronym{rmse}{RMSE}{root-mean-square error}
\newacronym{rss}{RSS}{received signal strength}
\newacronym{se}{SE}{spectral efficiency}
\newacronym{sdp}{SDP}{semidefinite programming}
\newacronym{sdr}{SDR}{semidefinite relaxation}
\newacronym{sinr}{SINR}{signal-to-interference-plus-noise ratio}
\newacronym{smse}{SMSE}{sum mean squared error}
\newacronym{sdma}{SDMA}{space-division multiple-access}
\newacronym{snr}{SNR}{signal-to-noise ratio}
\newacronym{soa}{SoA}{state-of-the-art}
\newacronym{sre}{SRE}{smart radio environment}
\newacronym{sir}{SIR}{signal-to-interference ratio}
\newacronym{toa}{ToA}{time-of-arrival}
\newacronym[plural=UEs, firstplural=users' equipment (UEs)]{ue}{UE}{user equipment}
\newacronym{ul}{UL}{uplink}
\newacronym{ula}{ULA}{uniform linear array}
\newacronym{upa}{UPA}{uniform planar array}
\newacronym{uatf}{UatF}{use-and-then-forget}
\newacronym{tdd}{TDD}{time-division duplex}
\newacronym{zf}{ZF}{zero-forcing}
\newacronym{mr}{MR}{maximum-ratio}

\newacronym{wrt}{w.r.t.}{with respect to}
\newacronym{sdn}{SND}{software defined network}
\newacronym{oran}{ORAN}{open RAN}
\newacronym{6g}{6G}{sixth generation}

\newacronym{pd}{PD}{power detector}

\definecolor{gold}{rgb}{0.85,.66,0}
\definecolor{amaranth}{rgb}{0.9, 0.17, 0.31}

\begin{document}

\title{
    Autonomous RISs and Oblivious Base Stations: \\ The Observer Effect and its Mitigation
}

\author{
    {Victor~Croisfelt},~\IEEEmembership{Member,~IEEE,}
    {Francesco~Devoti},~\IEEEmembership{Member,~IEEE,}
    {Fabio~Saggese},~\IEEEmembership{Member,~IEEE,}
    {Vincenzo~Sciancalepore},~\IEEEmembership{Senior Member,~IEEE,}
    {Xavier~Costa-P\'erez},~\IEEEmembership{Senior Member,~IEEE,} and
    {Petar~Popovski},~\IEEEmembership{Fellow,~IEEE}
    % <-this % stops a space
    \thanks{
    V. Croisfelt, F. Saggese, and P. Popovski are with Aalborg Universitet, 9220 Aalborg, Denmark. 
    F. Devoti and V. Sciancalepore are with NEC Laboratories Europe, 69115 Heidelberg, Germany.
    X. Costa-P\'erez is with i2cat, ICREA, and NEC Laboratories Europe, 08034 Barcelona, Spain. 
    This work was supported by the Villum Investigator grant "WATER" from the Velux Foundation, Denmark, by the EU H2020 RISE-6G project under grant agreement no. 101017011, and by the EU SNS JU INSTINCT project under grant agreement no. 101139161. Corresponding author email: \texttt{vcr@es.aau.dk}.
    }%
}

% The paper headers
% \markboth{Journal of \LaTeX\ Class Files,~Vol.~14, No.~8, August~2021}%
% {Shell \MakeLowercase{\textit{et al.}}: A Sample Article Using IEEEtran.cls for IEEE Journals}

% \IEEEpubid{0000--0000/00\$00.00~\copyright~2021 IEEE}
% Remember, if you use this you must call \IEEEpubidadjcol in the second
% column for its text to clear the IEEEpubid mark.

\maketitle

%\glsunset{ris}
%============================
\begin{abstract}
    Autonomous \glspl{ris} offer the potential to simplify deployment by reducing the need for real-time remote control between a \gls{bs} and an \gls{ris}. However, we highlight two major challenges posed by \emph{autonomy}. The first is \emph{implementation complexity}, as autonomy requires \glspl{hris} equipped with additional onboard hardware to monitor the propagation environment and perform local \gls{chest}, a process known as probing. The second challenge, termed \emph{probe distortion}, reflects a form of the observer effect: during probing, an \gls{hris} can inadvertently alter the propagation environment, potentially disrupting the operations of other communicating devices sharing the environment. Although implementation complexity has been extensively studied, probe distortion remains largely unexplored. To further assess the potential of autonomous \glspl{ris}, this paper comprehensively and pragmatically studies the fundamental trade-offs posed by these challenges collectively. In particular, we examine the \emph{robustness} of an \gls{hris}-assisted \gls{mmimo} system by considering its critical components and stringent conditions. The latter include: (a) two extremes of implementation complexity, represented by minimalist operation designs of two distinct \gls{hris} hardware architectures, and (b) an \emph{oblivious} \gls{bs} that fully embraces probe distortion. To make our analysis possible, we propose a \emph{physical-layer orchestration framework} that aligns \gls{hris} and \gls{mmimo} operations. We present empirical evidence that autonomous \glspl{ris} remain promising under stringent conditions and outline research directions to deepen probe distortion understanding.
\end{abstract}
%============================

%============================
% Preamble
%============================
\glsresetall 

\bstctlcite{IEEEexample:BSTcontrol} % Enable et.al in refs.

\glsunset{5g}
\glsunset{6g}
\glsunset{ris}

%============================
\section{Introduction}\label{sec:intro}
%============================
\noindent
\IEEEPARstart{R}{econfigurable} intelligent surfaces (\glspl{ris}) are an emerging technology with a significant role on the research agenda toward the \gls{6g}~\cite{Yang2016,DiRenzo2020,Pan2021}. An \gls{ris} consists of a grid of programmable elements that can dynamically control the reflection properties of incoming electromagnetic waves by adjusting the phase shifts of individual elements, collectively termed as a \textit{configuration}~\cite{Yang2016}. This technology envisions smart radio environments where multiple \glspl{ris} are deployed to possibly offer benefits such as enhanced \gls{se} and reduced electromagnetic-field exposure~\cite{DiRenzo2020,Strinati2021}. In this regard, most research has focused on \textit{nearly-passive} or \textit{solely-reflective} \glspl{ris}, which possess minimal hardware for element configuration and external communication~\cite{Yang2016,DiRenzo2020,Pan2021,Strinati2021}. Systems assisted by nearly-passive \glspl{ris} often operate within a centralized, \textit{non-autonomous} framework, where the \glspl{ris} are typically controlled by \glspl{bs} via \textit{dedicated} control channels and \textit{explicit} control signaling. This framework heavily relies on end-to-end \gls{chest} protocols to optimize \gls{ris} operations~\cite{Wang2020:ce}. Thus, achieving efficient real-time remote control poses a significant challenge to their practical implementation~\cite{Bjornson2020a,bjornson2021signalprocessing}. Especially,~\cite{Zappone2021,Saggese2023,Saggese2023:ris-mec} show that establishing and designing dedicated, explicit control can be detrimental to communication performance, leading to reduced \gls{se} gains and increased latencies. Notably, control costs arise from the allocation of physical resources, such as bandwidth and infrastructure, and engineering requirements that introduce control overhead and reliability issues. Control design also adds unnecessary complexity by requiring simultaneous consideration of multiple factors. These issues underscore a common oversight in prior studies, which often downplayed control-related costs and errors by indiscriminately assuming \textit{ideal control} conditions~\cite{Yang2016,DiRenzo2020,Pan2021,Strinati2021}.  

\begin{figure}[!t]
    \centering
    \vspace{0.05in}
    \includegraphics[width=0.99\columnwidth,height=4cm]{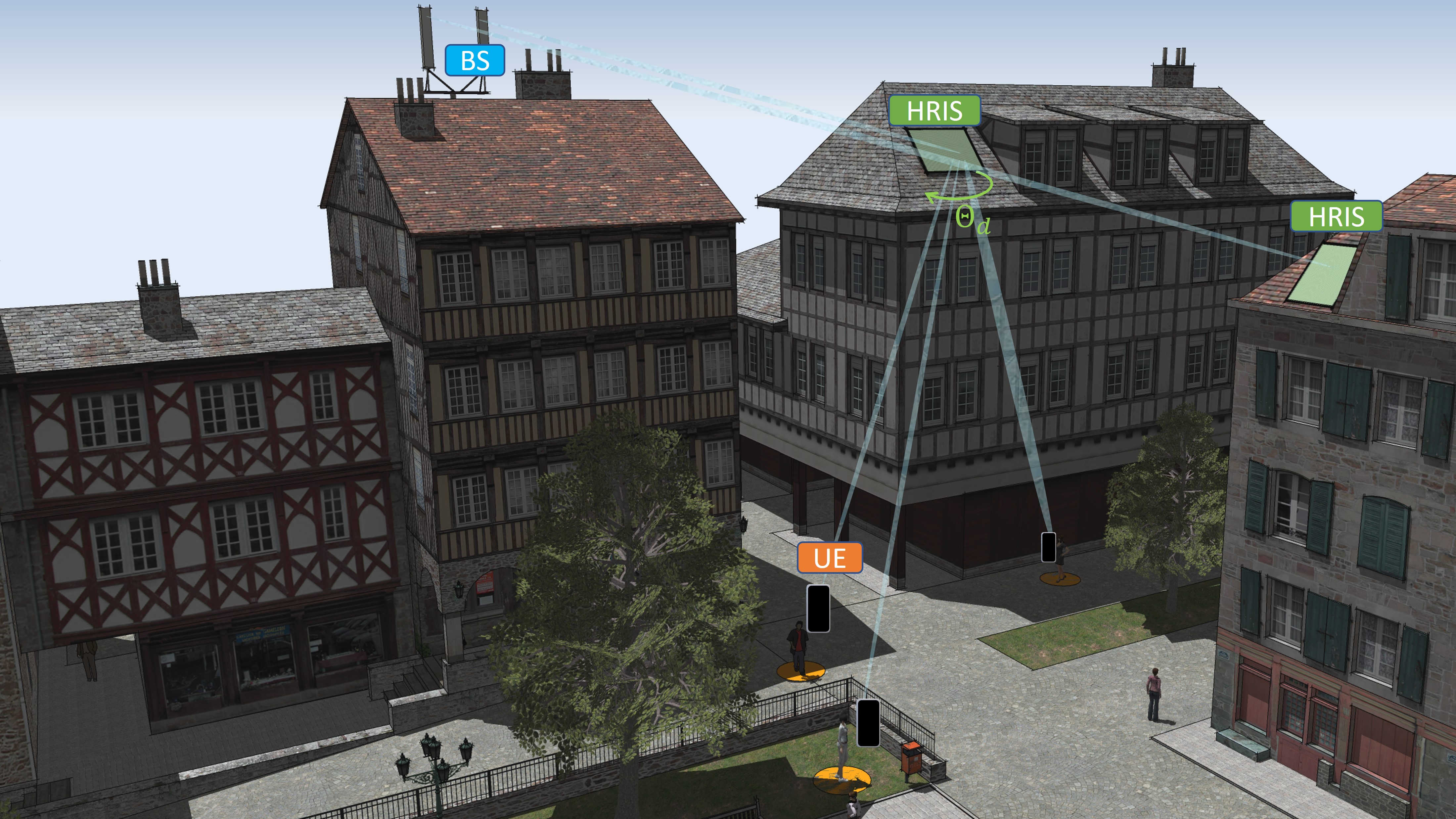}
    \caption{
        Open system models of autonomous RISs allow multiple HRISs to enhance communication performance between BSs and UEs without dedicated and explicit control.
    }
    \vspace{-0.3in}
    \label{fig:intro}
\end{figure}

To obviate the need for real-time remote control, recent works focused on studying decentralized frameworks with \emph{autonomous} \glspl{ris}, which operate independently of \glspl{bs} while bypassing dedicated and explicit control~\cite{Jian2022:hw-channels-survey}. This marks a paradigm shift from the traditional hierarchical \gls{bs}-\gls{ris} control to \textit{open} \gls{ris}-assisted system models, as illustrated in Fig.~\ref{fig:intro}. Inspired by the potentials of this alternative and the current uncertainty surrounding its feasibility~\cite{Subrt2010,alamzadeh2021reconfigurable,alexandropoulos2021hybrid,albanese2021marisa,albanese2024ares}, this paper aims to comprehensively and pragmatically examine the fundamental trade-offs in designing and deploying systems assisted by autonomous \glspl{ris}, focusing on two critical challenges: \textit{implementation complexity} and \textit{probe distortion}. To the best of our knowledge, this paper is the first to address both major challenges within a unified framework. Henceforth, we use the term ``autonomous \gls{ris}'' to refer to the underlying technology and ``\gls{hris}'' for the hardware that implements it. Conversely, ``non-autonomous \gls{ris}'' denotes the standard ``controlled \gls{ris}'' setup, where a \gls{bs} manages a nearly-passive \gls{ris} via dedicated, explicit control.

%%%%%
\subsection{The Two Major Challenges Posed by Autonomy}\label{sec:intro:challenges}
%%%%%

%
\paragraph*{\textbf{1. Implementation complexity}}
From a hardware perspective, autonomy relies on \glspl{hris} equipped with additional onboard hardware to monitor the propagation environment and perform local \gls{chest}, potentially increasing design complexity and associated costs per device. The term ``hybrid'' highlights their non-passive nature, allowing them to sense by simultaneously absorbing and reflecting incoming waves while lacking the capability to transmit signals; thus, positioning them between nearly-passive \glspl{ris} and relays~\cite{Bjornson2020a,Bjornson2020,Jian2022:hw-channels-survey}. Pioneering hardware solutions for \glspl{hris} are presented in~\cite{Subrt2010,alamzadeh2021reconfigurable,alexandropoulos2021hybrid,albanese2021marisa,albanese2024ares}, including additional components such as \gls{rf} chains and computing capabilities to execute \gls{dsp} methods. Notably, these works indicate that a minimal \gls{hris} implementation must alternate between \textbf{two operation modes}. $i)$ \textit{Probe mode:} The \gls{hris} actively probes the environment to detect the \glspl{bs} and \glspl{ue}, followed by a local \gls{chest} procedure of their \gls{csi}. The term ``probe'' highlights the \gls{hris}' active interaction with the propagation environment, distinguishing it from ``sense,'' which would suggest a more passive approach. $ii)$ \textit{Reflection mode:} Leveraging the probing knowledge, the \gls{hris} autonomously self-configures to assist ongoing communication performance.

We consider two \gls{hris} hardware architectures from~\cite{Jian2022:hw-channels-survey}: a low-complexity ``\gls{pd}-enabled'' and a more complex ``\gls{dsp}-enabled'' counterpart. These represent two extremes of implementation complexity concerning \gls{dsp} power. We design their respective probe and reflection modes to leverage their strengths while being mindful of their weaknesses. Notably, the aforementioned designs remain \textit{minimal} (strict), focusing on essential mathematical analysis rather than the ultimate optimization of \gls{hris} operations. Though not explicitly analyzed, minimal designs also promote low latencies in \gls{hris} operation, a highly desired feature.

\paragraph*{{\textbf{2. Probe distortion}}}
We note that autonomous \glspl{ris} can introduce a form of the \emph{observer effect}, a fundamental concept in physics stating that the act of observation inherently disturbs the observed system. In our context, the \gls{hris}' probing actions can alter the channel state, potentially disrupting the operations of other communicating devices sharing the environment. We term this disruption as \textit{probe distortion}, where ``distortion'' is defined as any alteration that modifies a signal's original shape or characteristics, without specifying whether the impact on communication performance is \textit{unfavorable} or \textit{favorable}. To our knowledge, this effect has often been overlooked in the literature, which arises primarily because current technology prevents the \gls{hris} from dynamically and seamlessly switching between fully absorbing and fully reflecting incoming waves~\cite{Subrt2010,alamzadeh2021reconfigurable,albanese2021marisa,albanese2024ares,alexandropoulos2021hybrid}. \textit{In essence, the higher the desired probing performance, the higher the level of probing distortion.}

Of particular importance, if unfavorable, probe distortion can be addressed in two main ways, depending on the \gls{bs}' awareness of the \gls{hris}—where the \gls{bs} often acts as the network coordinator~\cite{DiRenzo2020}. a) \textit{Informed BS:} The \gls{bs} is fully or partly aware of the \gls{hris} operation. Thus, the \gls{bs} can mitigate probe distortion by, for example, adopting a \textit{stop-and-wait} strategy, pausing its operation until probing concludes. This option incurs higher overhead, requiring the \gls{hris} to share information about its operation with the \gls{bs}, or for the \gls{bs} to actively monitor the environment to discern whether disturbances are due to the \gls{hris} or other causes, potentially wasting resources. Additionally, this information may need to be continuously updated due to the possible adaptive nature of the probe mode and the dynamics of the propagation environment. b) \textit{Oblivious BS:} The \gls{bs} is completely unaware of the \gls{hris} operation and executes its tasks carelessly.

We argue that considering an informed \gls{bs} presents a \textit{chicken-and-egg dilemma}, as the primary goal of autonomy is to minimize—\textit{ideally eliminate}—the need for dedicated, explicit control, much like the oblivious \gls{bs} scenario. The former also introduces higher complexity in network design and operational management, leading to higher resource consumption. Hence, we consider an oblivious \gls{bs} scenario—a highly stringent condition where the \gls{bs} fully embraces probe distortion, with no dedicated, explicit control over the \gls{hris}. While this scenario may not represent a definitive practical implementation, analyzing it is essential for risk assessment, providing insights into the consequences of completely lacking control upon a \gls{ris}, with significant academic and industrial implications. From an industrial point-of-view, an oblivious \gls{bs} means that no changes in a currently deployed \gls{bs} are needed to deploy an \gls{hris}. While our discussion focuses on the \gls{bs}' awareness, note that \glspl{ue} can also experience probe distortion—being its proper evaluation beyond our scope. For example, in carrier-sensing random access~\cite{Croisfelt2022,Croisfelt2023}, \glspl{ue} evaluate their channel qualities, which may be impacted by probe distortion, making them more likely to be ``oblivious'' due to their resource scarcity.

%%%%%
\subsection{Why Do We Need a PHY-Layer Orchestration Framework?}\label{sec:intro:motivation}
%%%%%
\noindent
Building on the \gls{5g} standard~\cite{3gpp_nr}, we consider an \gls{hris} assisting a \gls{mmimo} system with an oblivious \gls{bs}. Typically, an \gls{mmimo} system works in \gls{tdd} mode to limit the \gls{csi} acquisition overhead~\cite{massivemimobook}. This mode organizes the time-frequency resources in \textit{coherence blocks}, within which the channel remains time-invariant and frequency-flat. Each coherence block ranges from hundreds to several thousands of complex-valued samples, or \textit{samples} for short, depending on the physical characteristics of the propagation environment. The \gls{tdd} mode sequentially divides each coherence block into \textbf{two operation phases}~\cite{massivemimobook}. $1$) \textit{\Gls{chest} phase:} The \glspl{ue} transmit \gls{ul} pilot signals, or \textit{pilots} for short, to enable the \gls{bs} to perform \gls{chest} and obtain \textit{instantaneous} \gls{csi}. We often omit the \gls{ul} prefix from pilot-related quantities when it is clear from context. Due to channel reciprocity, the estimated \gls{csi} at the \gls{bs} side applies to both \gls{dl} and \gls{ul} directions. $2$) \textit{\Gls{comm} phase:} By using the estimated \gls{csi}, the \gls{bs} can compute \textit{spatial multiplexing} techniques (transmit precoding and receiver combining schemes). For simplicity, we assume that these computations do not incur any overhead. This phase comprises the transmission of \gls{dl} and \gls{ul} payload data while the \gls{bs} spatially separates the \glspl{ue}.

While the system operates through these phases within a given coherence block, the \gls{hris} must autonomously alternate between its two operation modes. This is where a \gls{phy}-layer orchestration framework comes into play. Such a framework must outline: (a) how the \gls{hris} operation modes are aligned with the simultaneous \gls{mmimo} operation phases; and (b) how an intelligent controller acting upon the \gls{hris} can assess the operation modes. To simplify the discussion, we often omit mentioning \gls{hris} and \gls{mmimo} about operation modes and phases, respectively, as the word ``mode'' always refers to \gls{hris} and ``phase'' to \gls{mmimo}.

%%%%%
\subsection{Contributions}\label{sec:intro:contributions}
%%%%%
\noindent
Our initial contribution is a \gls{phy}-layer orchestration framework, which builds the foundation for comprehensively and pragmatically investigating the following trade-offs concerning the above-stated autonomy challenges.\\ 
\indent (1) \textit{Implementation complexity trade-off:} We aim to understand how the overall \gls{hris} performance correlates with the two implementation complexity extremes, characterized by the \gls{pd}- and \gls{dsp}-enabled hardware architectures. Here, ``overall'' refers to both probing and reflecting performance. We also note that each hardware architecture operates differently, resulting in probe distortion with distinct characteristics.\\
\indent (2) \textit{Autonomous \gls{ris} trade-off:} We aim to evaluate the effects of both implementation complexity and probe distortion on the communication performance of an \gls{mmimo} system. Specifically, probe distortion gives rise to the \textit{autonomy paradox}, which suggests that the communication performance of an \gls{hris}-assisted \gls{mmimo} system can be worse than that of an equivalent, \textit{standalone} \gls{mmimo} system. This can occur because probe distortion can hinder spatial multiplexing at an oblivious \gls{bs}, as it relies on \gls{csi} affected by the probing distortion; while efforts to mitigate probe distortion can reduce reflecting performance, as the reflecting performance is inherently linked to the probing one (the output of the probe mode is the input of reflect mode, forming a cascaded system). Thus, our goal is to assess whether probe distortion is \textit{unfavorable} to communication performance. We refer to instances where \gls{hris}-assisted communication performance exceeds that of a standalone \gls{mmimo} as the \textit{robust feasibility region}.

We stress that our aim is \textit{not} to provide ultimate optimal design choices; rather, through minimalist designs and the consequent simplified mathematical analysis, we seek to comprehensively and pragmatically uncover the fundamental scaling rules of these trade-offs. Notably, we highlight that the degree of implementation complexity will be controlled by changing between the two hardware architectures. And, the level of probe distortion can be managed by adjusting the (relative) duration of the probe mode, and it also varies according to the hardware architecture.

Our numerical simulations show that \gls{hris}-assisted communication performance can outperform the standalone performance for a typical suburban setting with \glspl{ue} in cell-edge conditions. Intriguingly, probe distortion is observed to be \textit{dual} in the ability to be favorable or unfavorable to communication performance. This provides empirical evidence that autonomous \glspl{ris} can be a promising alternative for practical \glspl{ris} deployment, even under the considered stringent conditions; most impressively, completely lacking any form of dedicated and explicit control. 

%%%%%
\subsection{Paper Outline}
%%%%%
\noindent
Section~\ref{sec:related-work} reviews related work, while Section~\ref{sec:system-model} outlines the \gls{hris}-assisted \gls{mmimo} system model. In Section~\ref {sec:framework}, we introduce our orchestration framework. Section~\ref{sec:hris} and~\ref{sec:mmimo} detail the \gls{hris} operation modes and the \gls{mmimo} operation phases, respectively. Experiments and discussion are provided in Section~\ref{sec:results}, followed by the conclusions in Section~\ref{sec:conclusions}.

%%%%%
\subsection{Notation}
%%%%%
\noindent
Vectors and matrices are in bold lowercase and uppercase letters, respectively. The $i,j-$th element of a matrix $\mathbf{X}$ is $[\mathbf{X}]_{i,j}$; the $i-$th element of a vector $\mathbf{y}$ is $y_i$. The identity matrix of size $N$ is denoted as $\mathbf{I}_N$ while the vector or matrix of zeroes is $\mathbf{0}$, whose dimensions are specified by the context. Complex conjugate, transpose, Hermitian transpose, and diagonal matrix operators are denoted as $(\cdot)^*$, $(\cdot)^\transp$, $(\cdot)^\htransp$, and $\diag{(\cdot)}$, respectively. The $\ell_2-$norm is denoted as $\ltwonorm{\cdot}$, and, when convenient, the inner product between $\mathbf{x}$ and $\mathbf{y}$ is $\langle \mathbf{x}, \mathbf{y} \rangle$ while $\circ$ denotes the Hadamard product. Integer sets are represented by calligraphic letters, \textit{e.g.}, $\mathcal{A}$ with cardinality $|\mathcal{A}|\!=\!A$, whereas $\mathbb{N}$, $\mathbb{R}$ and $\mathbb{C}$ denote the sets of natural, real, and complex numbers, respectively. The operators $\Re(\cdot)$ and $\Im(\cdot)$ respectively return the real and the imaginary part of a number. The conditional \gls{pdf} is given by $p(x;E)$, for a random variable $x$ given an event $E$. The exponential distribution with parameter $\zeta$ is $\mathrm{Exp} (\zeta)$. The right-tail distribution of a central $\chi^2_n-$distributed random variable $x$ with $n$ degrees of freedom is $Q_{\chi^{2}_{n}}(x)$ while $Q_{\chi^{2}_{n}(\mu)}(x)$ represents a non-central one with non-centrality parameter $\mu$. The complex Gaussian distribution with mean $\mu$ and variance $\sigma^2$ is denoted as $\mathcal{CN}(\mu,\sigma^2)$. We use $\mathcal{O}(\cdot)$ for big-O notation. For clarity, we use the word ``channel'' to refer to channel vectors or matrices of channel responses. Other less frequent notations are clarified when needed.

%============================
\section{Related Work}\label{sec:related-work}
%============================
\noindent
Despite improvements in communication performance and innovative applications~\cite{Luo2021,he2021reconfigurable,albanese2022ris,wei2023wireless,li2023ris}, \glspl{ris} present significant challenges mainly related to their integration into network architecture, such as the execution of end-to-end \gls{chest}~\cite{DiRenzo2020,Pan2021,Strinati2021}. Methods to integrate non-autonomous \glspl{ris} into the network architecture have been introduced in the literature over the last few years; \textit{e.g.},~\cite{liaskos2018using} proposes a software-defined network approach, while~\cite{liaskos2020end} exploits machine learning in a similar setting. Technical challenges are further discussed in~\cite{liaskos2022software, strinati2021reconfigurable}. Notably, initial standardization efforts are underway to incorporate \glspl{ris} into \gls{6g} standards~\cite{ETSI-GR-RIS-001,ETSI-GR-RIS-003}, requesting further validations. While autonomous \glspl{ris} may reduce the need for network integration, our work focuses primarily on \gls{phy}-layer aspects and does not specifically address this issue.

The passive nature of non-autonomous \glspl{ris} complicates the end-to-end \gls{chest} process~\cite{DiRenzo2020,Pan2021,wei2021channel2}. This is further worsened by the large number of \gls{ris} elements, which increases the complexity of \gls{chest} and raises control overhead, eventually reducing the quality of the acquired \gls{csi}; that is, unfavorably leading to imperfect and/or outdated \gls{csi}~\cite{9366805}. Several distinct \gls{chest} procedures have been proposed for non-autonomous \gls{ris}-assisted systems~\cite{chen2023channel,zhang2023channel,fernandes2023channel,shen2023deep}. In~\cite{Ahmed2021,Zhi2022,Zhi2022a}, the performance of \gls{ris}-assisted \gls{mmimo} systems is analyzed under imperfect \gls{csi}, examining various precoding and combining techniques, as well as methods for optimizing \gls{ris} configurations using either instantaneous or statistical \gls{csi}. In~\cite{Hu2022}, the authors study the case of mobile \glspl{ue} with outdated \gls{csi}. \textit{However, in all these works, it is assumed that the \gls{bs} controls the \gls{ris} with negligible overhead and idealized precision.} This assumption is overoptimistic since it overlooks the challenges of designing a dedicated control channel and its potentially harmful effects on communication performance, as shown in~\cite{Zappone2021,Saggese2023,Saggese2023:ris-mec}.

Autonomous \glspl{ris} partially address the end-to-end \gls{chest} issue by redistributing the \gls{chest} tasks between the \gls{bs} and the \gls{hris}. This approach imposes additional costs on the \gls{hris} to reliably receive and process signals for local \gls{chest}~\cite{alamzadeh2021reconfigurable}, yet it shows significant potential, as motivated by~\cite{albanese2021marisa,albanese2024ares,alexandropoulos2021hybrid,schroeder2021passive}, leading to studies on local \gls{chest} procedures. For example,~\cite{taha2021enabling} employs a compressive sensing approach for local \gls{chest} relying only on a subset of \gls{hris} elements, while~\cite{schroeder2022channel} exploits \gls{ul} pilots for local \gls{chest}. However, their focus is on enhancing local \gls{csi} quality, overlooking other aspects. The closest works to ours are~\cite{albanese2021marisa,albanese2024ares,saigre2022self}, in which the \gls{hris} self-optimizes to assist ongoing communication. However, these works only consider the \gls{comm} phase assuming prior \gls{csi} knowledge, overlooking the effects of probe distortion. Our work provides a more comprehensive analysis encompassing both \gls{chest} and \gls{comm} aspects.

%============================
\section{System Model}\label{sec:system-model}
%============================
\noindent
Consider a single-cell \gls{mmimo} system where {an oblivious \gls{bs}} equipped with a \gls{ula} of $M$ antennas simultaneously serving $K$ single-antenna \glspl{ue} that are already scheduled, often referred to as \textit{scheduled \glspl{ue}} if the context demands.\footnote{
    Scheduling \glspl{ue} in the presence of an \gls{hris} is related to the \gls{ris}-assisted initial access problem~\cite{Croisfelt2022,Croisfelt2023}, and it is out of the scope of this paper. However, during scheduling, we ensure that \glspl{ue} have strong enough channels to the \gls{bs} so they can still be spatially separable. If \glspl{ue} were served only with the assistance of the \gls{hris}, spatial separability would be compromised since the channels become linearly dependent. This issue is a well-known problem in \gls{ris}-assisted \gls{mmimo} systems, \emph{e.g.}, see~\cite{Ahmed2021,Zhi2022,Zhi2022a}.
    \label{foot:active}
} We denote as $K_{\max}\!\geq\! K$ the maximum number of \glspl{ue} that can be supported by the system. Based on the Plug\&Play approach from~\cite{albanese2021marisa,albanese2024ares}, an \gls{hris} is deployed to autonomously enhance the propagation conditions. The \gls{hris} is comprised of $N\!=\! N_x N_z$ elements that are arranged as a \gls{upa}, where $N_x$ and $N_z$ denote the number of elements along the $x-$ and $z-$axis, respectively. We introduce the sets $\mathcal{M}$, $\mathcal{K}$, and $\mathcal{N}$ to index \gls{bs} antennas, \glspl{ue}, and \gls{hris} elements, respectively. The time-frequency domain is sliced into coherence blocks of $\tau_c$ samples, indexed by the set $\mathcal{T}_c$, where narrow-band wireless transmissions occur at a carrier frequency $f_c$ with wavelength $\lambda$ and bandwidth $B$. Fig.~\ref{fig:geometry} provides a geometric representation of the system.

%%%%%
\subsection{Basic HRIS Operation}\label{sec:hris-model}
%%%%%
\noindent
The \gls{hris} has the sensing capability to both absorb and reflect the incoming waves simultaneously. This can be realized through the use of directional couplers~\cite{alexandropoulos2021hybrid,alamzadeh2021reconfigurable}, whose \textit{coupling parameter} $\eta\!\in\![0,1]$ dictates the \emph{fixed} fraction of the received power from an incoming wave that is reflected into the environment; thus, the \textit{fraction of power absorbed} by the \gls{hris} is $1-\eta$.\footnote{
    We stress that, with the current technology~\cite{alexandropoulos2021hybrid}, the coupling parameter $\eta$ is set by the \gls{hris} hardware design and cannot be tuned dynamically after deployment; but it can be engineered during manufacturing to meet specific requirements of the propagation environment and intended applications.
    \label{foot:eta}
}
The \gls{hris} can alter the propagation environment by changing its configuration. Let $\mathbf{\Theta}\!=\!\mathrm{diag}([e^{j\theta_1},\dots,e^{j\theta_N}]^{\transp})$ be a configuration with $\theta_n\!\in\![0, 2\pi]$ denoting the phase-shift impressed by the $n-$th element. Due to directional couplers, both reflected and absorbed fractions are subject to $\mathbf{\Theta}$. Thus, the equivalent \gls{bs}-\gls{ue} channel for the $k-$th \gls{ue}, ${\mathbf{h}}_k\!\in\!\mathbb{C}^{M}$, is
\begin{equation}
    \mathbf{h}_k(\boldsymbol{\Theta}) = \mathbf{h}_{\mathrm{DR},k} + \mathbf{h}_{\mathrm{RR},k}(\boldsymbol{\Theta}),
    \label{eq:equivalent-channel}
\end{equation}
where $\mathbf{h}_{\mathrm{DR},k}\!\in\!\mathbb{C}^{M}$ is the \textbf{d}i\textbf{r}ect channel and $\mathbf{h}_{\mathrm{RR},k}\!\in\!\mathbb{C}^{M}$ is the \textbf{r}eflected channel, for $k\!\in\!\mathcal{K}$. Note that the reflected channel, $\mathbf{h}_{\mathrm{RR},k}$, and the equivalent channel, $\mathbf{h}_k$, are functions of the configuration, $\boldsymbol{\Theta}$, and can be written in terms of the \gls{hris}-\gls{ue} channel, $\mathbf{r}_k\!\in\!\mathbb{C}^{N}$, and the \gls{bs}-\gls{hris} channel, $\mathbf{G}\!\in\!\mathbb{C}^{M\times N}$. Below, we define the concept of a subblock to help us define how the \gls{hris} operates.

\begin{definition}[Subblock]
    We let a \textit{subblock} be a group of samples within the same coherence block. We denote as $\mathcal{T}\!\subseteq\!\mathcal{T}_{c}$ a subblock. Subblocks are indexed by $s$, which takes values from an index set $\mathcal{S}$ that indexes partitions of samples of size $|\mathcal{T}|$ from $\mathcal{T}_c$. In the special case that a subblock comprises a single sample, we have $s\!\in\!\mathcal{T}_c$ since $|\mathcal{T}|\!=\!1$.
    \label{def:subblock}
\end{definition}
\begin{assumption}[HRIS configuration change] 
    We assume that the \gls{hris} can change its configuration $\boldsymbol{\Theta}$ on a subblock basis. We denote as $\boldsymbol{\Theta}[s]$ the configuration impressed by the \gls{hris} at the $s-$th subblock, for $s\!\in\!\mathcal{S}$. 
    \label{assu:hris-configuration-change}
\end{assumption}

The above assumption aligns with the current technology~\cite{alexandropoulos2021hybrid,alamzadeh2021reconfigurable}. Indeed, an \gls{hris} requires a time ranging from microseconds to milliseconds to change its configuration~\cite{Dai2020}, whose exact value depends on how the \gls{hris} is built and might correspond to the duration of some samples~\cite{Saggese2023}. As a consequence of this assumption, the equivalent channel also changes on a subblock basis and eq.~\eqref{eq:equivalent-channel} can be rewritten as
\begin{equation}
    \mathbf{h}_k[s] = \mathbf{h}_{\mathrm{DR},k} + \mathbf{h}_{\mathrm{RR},k}[s],
    \label{eq:equivalent-channel-over-samples}
\end{equation}
where $\mathbf{h}_{\mathrm{DR},k}$ is not affected by the configuration change. As explained in Section~\ref{sec:intro:challenges}, the \gls{hris} transitions between two operation modes. We will consider that in each mode the \gls{hris} uses different configurations $\boldsymbol{\Theta}$, which are further specified in Section~\ref{sec:hris}. This results in distinct equivalent channels, referred to as the \textit{probing equivalent channel} in probe mode, denoted as ${\mathbf{h}}_{\mathrm{P},k}\!\in\!\mathbb{C}^{M}$, and the \textit{reflecting equivalent channel} in reflection mode, denoted as ${\mathbf{h}}_{\mathrm{R},k}\!\in\!\mathbb{C}^{M}$, for $k\!\in\!\mathcal{K}$.

\begin{figure}[!t]
    \centering
    \vspace{0.5mm}
    \includegraphics[width=8cm]{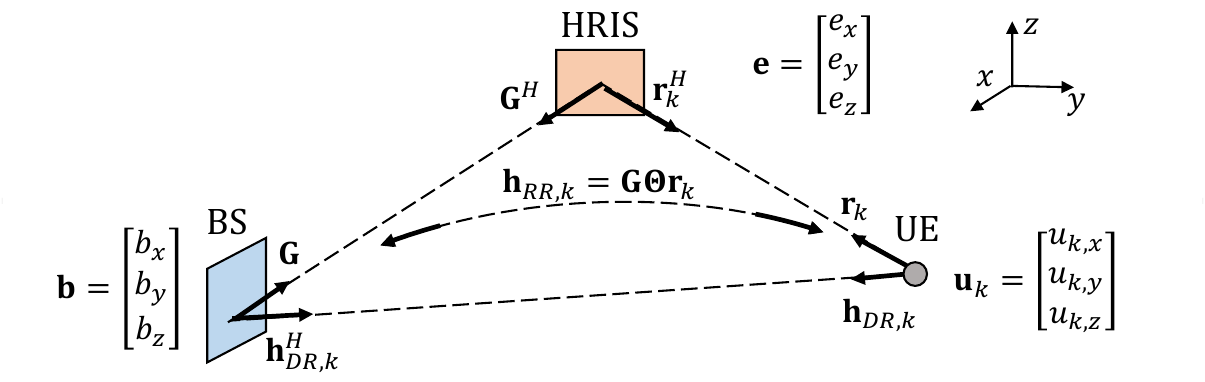}
    \vspace{-6pt}
    \caption{
        Geometric representation of the \gls{hris}-assisted \gls{mmimo} system, illustrating the \gls{bs}, \gls{hris}, and \gls{ue}, with channel notation defined for the \gls{ul} direction.
    }
    \vspace{-6mm}
    \label{fig:geometry}
\end{figure}
%

%%%%%
\subsection{Channel Models}\label{sec:channel-model}
%%%%%
\noindent
We assume a block-fading model~\cite{massivemimobook}. To simplify, we consider a single \gls{ue} $k\in\mathcal{K}$, a single coherence block, and, we also get rid of the $[s]$ notation in this subsection. Denote as $\mathbf{b}\in\mathbb{R}^{3}$, $\mathbf{e} \in \mathbb{R}^{3}$, and $\mathbf{u}_k \in \mathbb{R}^{3}$ the locations of the \gls{bs} center, the \gls{hris} center, and the $k-$th \gls{ue}, respectively. The position of the $m-$th \gls{bs} antenna is $\mathbf{b}_m\in\mathbb{R}^3$, for $ m\in\mathcal{M}$, while of the $n-$th \gls{hris} element is $\mathbf{e}_n\in\mathbb{R}^3$, for $n\in\mathcal{N}$. The inter-antenna and inter-element distances are set to $\lambda/2$. Let $\mathbf{a}_{\rm B}(\mathbf{p})\!\in\!\!\mathbb{C}^{M}$ and $\mathbf{a}_{\rm H}(\mathbf{p})\!\in\!\mathbb{C}^{N}$ denote the respective \gls{bs}' and \gls{hris}' array response vectors toward a generic location $\mathbf{p}\in\mathbb{R}^3$. The $n-$th element of $\mathbf{a}_{\mathrm{H}}(\mathbf{p})$ is~\cite{albanese2021marisa}
\begin{equation}
    [\mathbf{a}_{\rm H}(\mathbf{p})]_{n} = e^{j\langle\mathbf{k}(\mathbf{p},\mathbf{e}) , (\mathbf{e}_n - \mathbf{e})\rangle}, \text{ with } \mathbf{k}(\mathbf{p},\mathbf{e}) = \tfrac{2\pi}{\lambda}\tfrac{\mathbf{p}-\mathbf{e}}{\ltwonorm{\mathbf{e} - \mathbf{p}}}
    \label{eq:array-gain}
\end{equation}
being the wave vector; the vector $\mathbf{a}_\mathrm{B}(\mathbf{p})$ is derived similarly with $\mathbf{b},\mathbf{b}_m$ instead of $\mathbf{e},\mathbf{e}_n$, respectively. Next, the pathloss model between two generic locations $\mathbf{p},\mathbf{q}\!\in\!\mathbb{R}^3$ is~\cite{albanese2021marisa}: $ \gamma(\mathbf{p},\mathbf{q})\! =\! \gamma_0 ( {d_0}/{\ltwonorm{\mathbf{p} - \mathbf{q}}} )^\beta,$ where $\gamma_0$ is the channel power gain at a reference distance $d_0$ and $\beta$ is the pathloss exponent. In particular, we assume that the direct \gls{bs}-\glspl{ue} channels are under a pathloss exponent of $\beta_{\rm B}$, while the \gls{bs}-\gls{hris} and \gls{hris}-\gls{ue} are subject to $\beta_{\rm H}$. This assumption is reasonable, as the \gls{hris} is typically positioned to provide clearer propagation paths to the \gls{bs} and \glspl{ue}, with fewer obstructions compared to the \gls{bs}-\glspl{ue} paths~\cite{Xu2021}.
 
We assume an \gls{iid} Rician fading model for the \gls{bs}-\gls{ue} channel, $\mathbf{h}_{\mathrm{DR},k}$, and for the \gls{hris}-\gls{ue} channel, $\mathbf{r}_k$, while the \gls{bs}-\gls{hris} channel, $\mathbf{G}$, is \gls{los} dominant, and hence deterministic. The latter is valid if the \gls{hris} is deployed to have a strong \gls{los} toward the \gls{bs}, which is often the case due to the flexibility of deployment of the \gls{hris}~\cite{Xu2021}. However, no such assumption is made on the links between the \gls{bs}/\gls{hris} and \glspl{ue}, \textit{e.g.}, due to faster dynamics~\cite{massivemimobook}. Thus, we have
\begin{equation}
    \mathbf{h}_{\mathrm{DR},k}\sim\mathcal{CN}(\bar{\mathbf{h}}_{\mathrm{DR},k},\sigma^2_{{\mathrm{DR}}}\mathbf{I}_{M})\text{ and }\mathbf{r}_k\sim\mathcal{CN}(\bar{\mathbf{r}}_{k},\sigma^2_{{\mathrm{RR}}}\mathbf{I}_{N}),
    \label{eq:ricean-model}
\end{equation}
where $\bar{\mathbf{h}}_{\mathrm{DR},k}$ and $\bar{\mathbf{r}}_{k}$ are the \gls{los} components, while $\sigma^2_{\rm DR}$ and $\sigma^2_{\rm RR}$ are the relative powers of the \gls{nlos} components for the \gls{bs}-\gls{ue} and \gls{hris}-\gls{ue} channels, respectively. Based on the above, the \gls{los} components are $ \bar{\mathbf{h}}_{\mathrm{DR},k}\!=\!\sqrt{\gamma (\mathbf{b},\mathbf{u}_k)} \mathbf{a}_{\rm B}{(\mathbf{u}_k)},\,\bar{\mathbf{{r}}}_k\!=\!\sqrt{\gamma(\mathbf{u}_k,\mathbf{e})}\mathbf{a}_{\rm H}(\mathbf{u}_k), \text{ and } \mathbf{G}\!=\!\sqrt{\gamma(\mathbf{b},\mathbf{e})}\mathbf{a}_{\rm B}(\mathbf{e})\mathbf{a}_{\rm H}^{\htransp}(\mathbf{b})$. Accordingly, the reflected channel is 
\begin{equation}
    \mathbf{h}_{\mathrm{RR},k}\!=\!(\sqrt{\eta}\mathbf{G}\mathbf{\Theta}\mathbf{{r}}_k)
    \!\sim\! \mathcal{CN}\!\left(\sqrt{\eta}\mathbf{G}\mathbf{\Theta}\bar{\mathbf{r}}_k, \eta\!\,\gamma(\mathbf{b},\mathbf{e}) N \sigma^2_{{\mathrm{RR}}} \mathbf{Q} \right)\!,
    \label{eq:channel:reflected}
\end{equation}
where $\mathbf{Q}=\mathbf{a}_{\rm B}(\mathbf{e})\mathbf{a}_{\rm B}(\mathbf{e})^{\htransp}$ is a covariance matrix with ones in the diagonal and off-diagonal elements capturing the \gls{bs} antenna correlation evaluated at the \gls{hris} center. The equivalent \gls{bs}-\gls{ue} channel can be obtained by substituting eq.~\eqref{eq:channel:reflected} into~\eqref{eq:equivalent-channel-over-samples}.

%============================+
\section{A PHY-Layer Orchestration Framework}\label{sec:framework}
%============================
\noindent
In this section, we present our \gls{phy}-layer orchestration framework. We begin by introducing two design rules that underpin the framework, illustrated in Fig.~\ref{fig:framework-diagram}. Then, we provide a detailed presentation of the proposed framework, including the basic mathematical notation and underlying assumptions.
 
%%%%%
\subsection{The Two Design Rules}
%%%%%
\noindent
As motivated in Section~\ref{sec:intro}, we consider an \gls{hris}-assisted \gls{mmimo} system with an oblivious \gls{bs}, completely lacking dedicated and explicit control between the \gls{bs} and \gls{hris}. On this basis, we allow for \textit{minimal and implicit} control information exchange between the \gls{bs} and the \gls{hris} over \textit{existing} control channels. Specifically, the \gls{hris} can listen to standardized control channels—such as the physical downlink control channel (PDCCH)~\cite{3gpp_nr}—to acquire synchronization and data frame details, allowing it to align its operation modes to the \gls{bs}' \gls{chest} and \gls{comm} operation phases, similarly to a standard \gls{ue}. To effectively benefit from the \gls{hris} deployment, we propose an orchestration framework that pragmatically arranges the concurrent operation modes and phases within a coherence block at the \gls{phy} layer. This framework is structured around \textit{two design rules}, illustrated in Fig.~\ref{fig:framework-diagram} and detailed below.

\begin{figure}[!t]
    \centering
    \includegraphics[width=0.9\columnwidth]{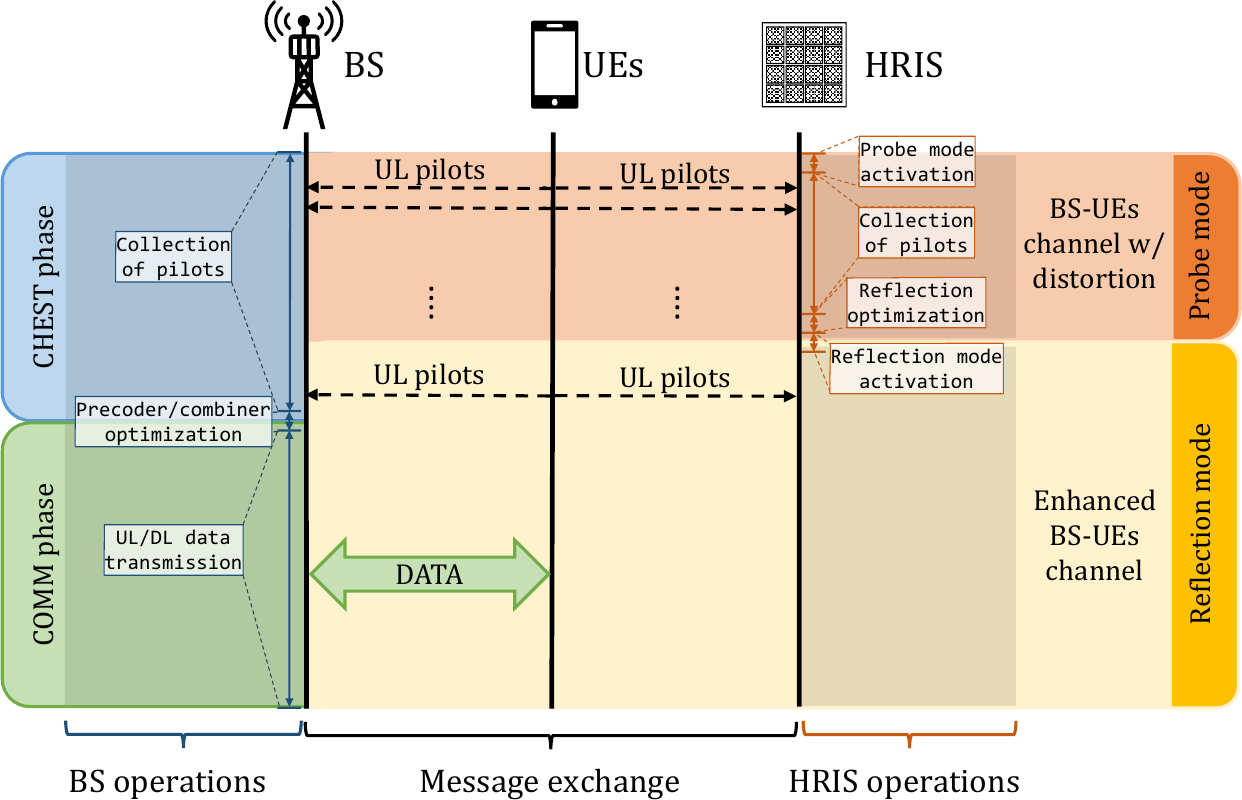}
    \vspace{-0.05in}
    \caption{
        Temporal evolution of the proposed \gls{phy}-layer orchestration framework within a coherence block. The \gls{mmimo} system alternates between two operation phases: $1$) \acrfull{chest} and $2$) \acrfull{comm}; while the \gls{hris} autonomously alternates between two operation modes: $i$) probe and $ii$) reflection.
    }
    \vspace{-0.2in}
    \label{fig:framework-diagram}
\end{figure}
\paragraph*{\textbf{First design rule}}
\textit{The probe mode must take place during the \gls{chest} phase.} This is a natural choice that enables the \gls{hris} to leverage \gls{ul} pilots for identifying scheduled \glspl{ue} and locally estimating their \gls{csi}, as in~\cite{schroeder2022channel}. Effectively, the \gls{hris} can exploit the channel reciprocity inherent from \gls{tdd} operation and, consequently, it can estimate the local \gls{csi} in the \gls{ul} direction only and extrapolate it to the \gls{dl}.\footnote{
    We assume channel reciprocity is perfectly achieved, \emph{e.g.}, by using carefully designed hardware and calibration algorithms~\cite{massivemimobook}, allowing us to focus on discussing our main ideas. Future research could explore what happens if channel reciprocity is violated by/at the \gls{hris}.
} The big downside of this design rule is that probing can alter the channel state during the \gls{chest} phase (observer effect), hence distorting the \gls{csi} estimated by the \gls{bs}. As mentioned, we name this effect as probe distortion. \textit{Therefore, probe distortion manifests as a distortion introduced by the \gls{hris} into the estimated \gls{csi} at the \gls{bs}.} This results in an imperfect, probe-distorted \gls{csi} at the \gls{bs}, which can adversely affect the spatial separation of \glspl{ue} during the \gls{comm} phase, potentially degrading communication performance.

\paragraph*{\textbf{Second design rule}}
\textit{The reflection mode must take place before the end of the \gls{chest} phase and during the entire \gls{comm} phase.} This approach attempts to address a key limitation of autonomous \glspl{ris}, which challenges a foundational principle of the \gls{mmimo} technology: the assumption that \gls{csi} estimated during the \gls{chest} phase remains consistent with the channel state during the \gls{comm} phase~\cite{massivemimobook}. To exemplify this, Fig.~\ref{fig:hris_effect} illustrates the evolution of the power of an equivalent \gls{ul} \gls{bs}-\gls{ue} channel, as defined in~\eqref{eq:equivalent-channel-over-samples}, over a coherence block, reproducing the \gls{hris}' switching between its operation modes. Distinct channel state characteristics are observed: during probe mode, the \textit{probing equivalent channels} can \textit{vary} as the \gls{hris} can alter its configuration to probe for \glspl{ue}. In contrast, during reflection mode, the \textit{reflecting equivalent channel} is \textit{stable} since the \gls{hris} loads and maintains a fixed reflection configuration after finishing probing, which is kept until the next coherence block begins. We assume that the computation of configurations does not incur any overhead. By imposing the start of the reflection mode to occur during the \gls{chest} phase, we aim to enable the \gls{bs} to collect enough samples of the reflecting equivalent channel, but on the effect of probe distortion, attempting to ensure adequate spatial separation of the \glspl{ue} during the \gls{comm} phase.\footnote{
    The design choices we made form \textit{one possible} orchestration framework. Alternative frameworks could be proposed, but we argue that our choices are both natural and well-aligned with \gls{mmimo} technology, providing a basic platform to analyze the relevant trade-offs outlined in Section~\ref{sec:intro:contributions}.
}

%%%%%
\subsection{Detailed Description}\label{sec:framework:detailed}
%%%%%
\noindent
Figure~\ref{fig:frame-design} illustrates how the coherence block is sliced simultaneously into the different operation phases and modes. We let $\tau_\chest$ and $\tau_\comm$ be the number of samples comprising the \gls{chest} and \gls{comm} phases, respectively, such that $\tau_c\!=\!\tau_\chest + \tau_\comm$. The \gls{comm} phase can be further divided into $\tau_d$ and $\tau_u$ samples for \gls{dl} and \gls{ul} data traffic, respectively; that is, $\tau_\comm\!=\!\tau_d + \tau_u$. Simultaneously, we let $\tau_\prob\leq\tau_\chest$ and $\tau_\refl$ be the number of samples comprising the probe and reflection modes, respectively, with $\tau_c\!=\!\tau_\prob + \tau_\refl$. 

We now outline the basic execution of the \gls{chest} phase via \gls{ul} pilot signaling~\cite{massivemimobook}. During connection establishment within a given coherence block, the \gls{bs} performs a pilot assignment $p(i)\!:\!\mathcal{K}\!\mapsto\!\mathcal{T}_p$, where each scheduled \gls{ue} is deterministically assigned a pilot from a total of $\tau_p$ pilots, indexed by the set $\mathcal{T}_p\!\subset\!\mathcal{T}_c$ with $|\mathcal{T}_p|\!=\!\tau_p$, for $i\!\in\!\mathcal{K}$. In other words, $p(i)\!\in\!\mathcal{T}_p$ represents the index of the pilot assigned to \gls{ue} $i$. We say a pilot is \textit{active} if it is assigned to a \gls{ue}; otherwise, it is \textit{inactive}. Note that at most only one \gls{ue} can be associated with each pilot. Each pilot $\boldsymbol{\upphi}_t\!\in\!\mathbb{R}^{\tau_p}$ spans for $\tau_p$ samples, for $t\!\in\!\mathcal{T}_p$. The pilots are selected from a \textit{pilot codebook} $\boldsymbol{\Upphi}\!\in\!\mathbb{R}^{\tau_p\times\tau_p}$. To avoid interference and simplify the analysis, we assume the following about the pilot codebook. 
\begin{assumption}[Orthogonal \gls{ul} pilot codebook]
    The pilot codebook contains mutually orthogonal pilots, such that $\boldsymbol{\upphi}_t^{\htransp} \boldsymbol{\upphi}_{t^\prime}\!=\!{\tau_p}$ if $t\!=\!t^\prime$ and $\boldsymbol{\upphi}_t^{\htransp} \boldsymbol{\upphi}_{t^\prime}\!=\!0$ if
    $t\!\neq\!t^\prime$, $\forall t,t^\prime\!\in\!\mathcal{T}_p$. In particular, we assume $\bm{\Upphi}\!=\! \sqrt{\tau_p} \mathbf{I}_{\tau_p}$ and that the maximum number of \glspl{ue} is equal to the pilot length, \emph{i.e.}, $K_{\max} = \tau_p$.
    \label{assu:ortho-pilots}
\end{assumption}
\begin{figure}[!t]
    \centering
    % This file was created with tikzplotlib v0.10.1.
\begin{tikzpicture}

\definecolor{darkorange25512714}{RGB}{255,127,14}
\definecolor{forestgreen4416044}{RGB}{44,160,44}
\definecolor{lavender233}{RGB}{233,233,233}
\definecolor{steelblue31119180}{RGB}{31,119,180}

\def\shift{-0.8cm}
\def\sep{1.2cm}
\def\vside{2.4cm}

\begin{axis}[
    % General format
    height=0.6in,
    width=2.7in,
    scale only axis,
    tick align=outside,
    tick pos=left,
    %legend pos=south east,
    % ylabel shift={
    %     -6pt
    % },
    legend style={
        draw=none,
        fill=none,
        at={(1.05,1)},
        anchor=south east,
        legend cell align=center,
        align=center
    },
    legend columns=-1,
    % unbounded coords=jump,
    xlabel={
        Subblocks, $s$
    },
    xmajorgrids,
    xmin=0, xmax=32,
    xtick style={
        color=black
    },
    xtick={0, 4, 8, 12, 16, 20, 24, 28, 32
    },
    % xticklabels={
    %     \(\displaystyle {0.0}\),
    %     %\(\displaystyle {0.1}\),
    %     \(\displaystyle {0.2}\),
    %     %\(\displaystyle {0.3}\),
    %     \(\displaystyle {0.4}\),
    %     %\(\displaystyle {0.5}\),
    %     \(\displaystyle {0.6}\),
    %     %\(\displaystyle {0.7}\),
    %     \(\displaystyle {0.8}\),
    %     %\(\displaystyle {0.9}\),
    %     \(\displaystyle {1.0}\)
    % },
    ymajorgrids,
    ylabel={
        \parbox[c]{1in}{
            Equivalent BS-UE\\channel gain [dB]
        }
    },
    ylabel shift={
         -8pt
    },
]
\addplot[black, very thick] 
table { %
    1.0000 -111.6496
    2.0000 -111.6496
    3.0000 -111.6496
    4.0000 -111.6496
    5.0000 -111.6496
    6.0000 -111.6496
    7.0000 -111.6496
    8.0000 -111.6496
    9.0000 -111.6496
    10.0000 -111.6496
    11.0000 -111.6496
    12.0000 -111.6496
    13.0000 -111.6496
    14.0000 -111.6496
    15.0000 -111.6496
    16.0000 -111.6496
    17.0000 -111.6496
    18.0000 -111.6496
    19.0000 -111.6496
    20.0000 -111.6496
    21.0000 -111.6496
    22.0000 -111.6496
    23.0000 -111.6496
    24.0000 -111.6496
    25.0000 -111.6496
    26.0000 -111.6496
    27.0000 -111.6496
    28.0000 -111.6496
    29.0000 -111.6496
    30.0000 -111.6496
    31.0000 -111.6496
    32.0000 -111.6496
};
% You can also add a label for the plot
\addlegendentry{Standalone}

% Plotting y = x^2
\addplot[red, very thick] 
table { %
    1.0000 -76.5429
    2.0000 -78.8144
    3.0000 -90.4589
    4.0000 -75.9378
    5.0000 -89.1139
    6.0000 -83.5829
    7.0000 -80.9906
    8.0000 -83.9178
    9.0000 -92.9616
    10.0000 -82.5230
    11.0000 -83.0026
    12.0000 -86.8323
    13.0000 -83.8642
    14.0000 -82.1091
    15.0000 -93.2753
    16.0000 -82.5394
};
% You can also add a label for the plot
\addlegendentry{Probing eq. ch., $\mathbf{h}_{\mathrm{P},k}[s]$}

\addplot[blue, very thick] 
table { %
    17.0000 -57.2316
    18.0000 -57.2316
    19.0000 -57.2316
    20.0000 -57.2316
    21.0000 -57.2316
    22.0000 -57.2316
    23.0000 -57.2316
    24.0000 -57.2316
    25.0000 -57.2316
    26.0000 -57.2316
    27.0000 -57.2316
    28.0000 -57.2316
    29.0000 -57.2316
    30.0000 -57.2316
};
% You can also add a label for the plot
\addlegendentry{Reflecting eq. ch., $\mathbf{h}_{\mathrm{R},k}[s]$}

% Draw a point to indicate the discontinuity
\addplot[only marks, mark=o, mark options={scale=1, color=red}] 
table { %
    16.0000 -82.5394
};

\addplot[only marks, mark=o, mark options={scale=1, color=blue}] 
table { %
    17.0000 -57.2316
};

%\fill[red] (-1, 0) circle (2pt);

% \addplot[name path=A, domain=0:5, samples=100] {x^2};
% \addplot[name path=B, domain=0:5, samples=100] {x+2};

% % Shade the area between the two curves
% \addplot[fill=blue!20] fill between[of=A and B];

% % Add a text label in the middle of the shaded area
% \node at (axis cs:2.5, 8) {Shaded Area Text};

\end{axis}
\end{tikzpicture}
    \vspace{-6pt}
    \caption{
        Example of the evolution of the equivalent \gls{ul} \gls{bs}-\gls{ue} channel gain, $\mathbf{h}_k[s]$ in~\eqref{eq:equivalent-channel-over-samples}, over a coherence block of $32$ subblocks with the \gls{hris} changing its configuration every subblock. During the probe mode, the channel state of the \textit{probing equivalent channels} can vary significantly whereas the \textit{reflecting equivalent channel} remains stable during the reflection mode. The oblivious \gls{bs} attempts to estimate the reflecting equivalent channel while the \gls{hris} is probing; as a result, \textit{probe distortion} can degrade the quality of the \gls{csi} at the \gls{bs}.
    }
    \label{fig:hris_effect}
    \vspace{-0.2in}
\end{figure}

The above is based on the rule of thumb described in~\cite{massivemimobook} for selecting the number of pilots without interference. We further assume that the \gls{hris} knows $\boldsymbol{\Upphi}$ and, hence, $\tau_p$ and $K_{\max}$, \emph{e.g.}, by listening to the PDCCH~\cite{3gpp_nr}. Due to our design rules and orthogonality, an issue emerges if the duration of the \gls{chest} phase is equal to the pilot length, that is, $\tau_\chest=\tau_p$. To see it, consider the following example.

\begin{example}
    Consider that $K\!=\!K_{\max}\!=\!\tau_p\!=\!2$ and that $\tau_\chest=\tau_p$. Assume that \gls{ue}-$1$ is assigned to the \gls{ul} pilot $[\sqrt{2},0]^\transp$ and \gls{ue}-$2$ to $[0,\sqrt{2}]^\transp$. Based on our framework, we want $0<\tau_\prob\!<\!\tau_\chest$; hence, we choose $\tau_\prob\!=\!1$. In this case, the \gls{hris} would receive just the first entries of the pilots. Since the first entry of \gls{ue}-$2$'s pilot is $0$, the \gls{hris} would be able to probe only \gls{ue}-$1$, no matter what \gls{ue}-$2$ does.
    \label{ex:no-detection}
\end{example}

To solve the above problem, we assume the following \textit{pilot repetition strategy}, consequently defining the duration of the \gls{chest} phase.\footnote{
    We note that the only explicit modification made in this work to incorporate autonomous \glspl{ris} into standard \gls{mmimo} technology is the repetition of \gls{ul} pilots. While this does not dictate practical implementation—such as using non-orthogonal pilot codebooks to eliminate the need for repetition—orthogonality simplifies the required designs and the interpretation of relevant trade-offs. End-to-end \gls{chest} procedures with non-autonomous \glspl{ris} also modify standard \gls{mmimo} (\textit{e.g.},~\cite{wei2021channel2}), making our assumption reasonable.
}

\begin{assumption}[\gls{ul} pilot repetition: Duration of the \gls{chest} phase]
    Each \gls{ue} re-transmits its pilot for $L>1$ times such that $\tau_\chest\!=\!L\tau_p$. We refer to each of the pilot repetitions as a \gls{ul} pilot subblock, following Definition~\ref{def:subblock}, which is indexed by the set $\mathcal{L}$ with $\mathcal{L}$ being a partition of the set $\mathcal{T}_p$. We index variables that occur on a pilot-subblock basis by introducing an $[l]$ in front of it, with $l\in\mathcal{L}$.
    \label{assu:pilot-repetition}
\end{assumption}

This assumption allows us to effectively accommodate the probe mode within the \gls{chest} phase while avoiding the problem seen in Example~\ref{ex:no-detection}. To enhance clarity, we will omit the prefix term ``pilot'' from subblocks when the context allows. We now define the duration of the modes as follows.

\begin{definition}[Duration of the modes]
    The probing duration can be defined as an integer multiple of the pilot length $\tau_p$, satisfying $0<\tau_\prob\!\leq\!\tau_\chest$. Thus, the probe mode spans for $\tau_\prob\!=\!C\tau_p$, where $1\!\leq\!C\!\leq\! L$ represents the number of pilot subblocks utilized by the \gls{hris} for probing. The specific subblocks during which the \gls{hris} probes are collected in the subset $\mathcal{C}\subseteq\mathcal{L}$. To clarify, we introduce a notation $[c]$ to index variables that occur on a pilot subblock basis during probing, with $c\in\mathcal{C}$. Hence, the fraction of the coherence block that the \gls{hris} operates in reflection mode is $\tau_\refl\!=\!\tau_c-C\tau_p$.
    \label{assu:probing-time}
\end{definition}

We further define the following relative quantities.

\begin{definition}[Relative duration of the modes within the CHEST phase]
    The relative duration of the probe mode within the \gls{chest} phase can be defined as:
    \begin{equation}
        \varpi = \frac{C}{L}, \text{ with } 0\leq\varpi\leq1,
        \label{eq:varpi:definition}
    \end{equation}
    where $\varpi$ equals $0$ in the absence of the probe mode, $C=0$, and equals $1$ when the probe mode occupies the entire duration of the \gls{chest} phase, $C=L$. Hence, the relative duration of the reflection mode is $1-\varpi$.
    \label{definition:varpi}
\end{definition}

With the orchestration framework and associated notation established, we can proceed to the system design, naturally dividing it into the \gls{hris} and \gls{mmimo} components, while remaining mindful of the trade-offs outlined in Section~\ref{sec:intro:contributions}.

\begin{figure}[!t]
    \centering
    \includegraphics[width=0.9\columnwidth]{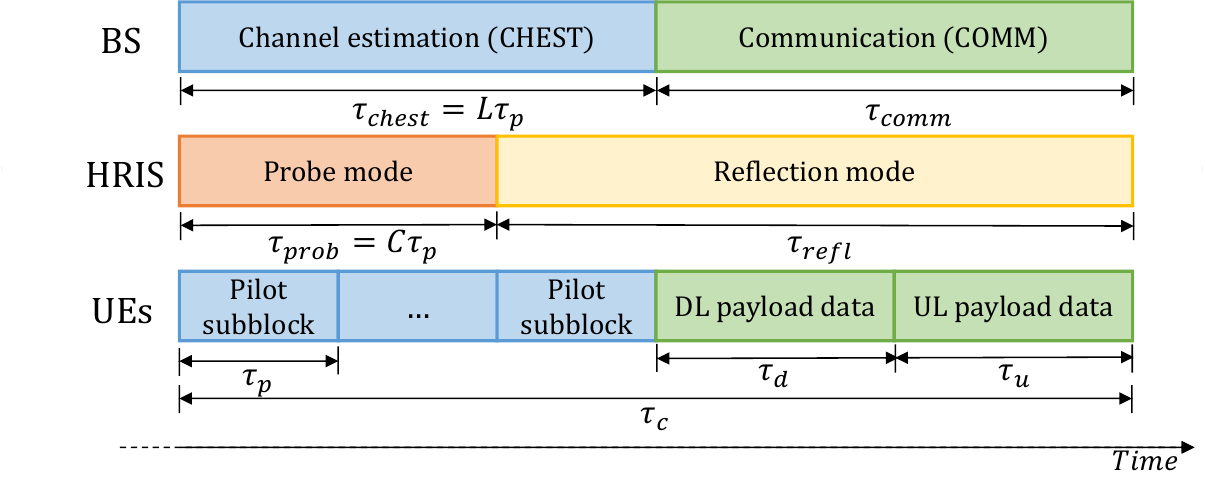}
    \vspace{-0.1in}
    \caption{
        Sample-based organization of a coherence block.
    }
    \label{fig:frame-design}
\end{figure}
%

%============================
\section{Designing the HRIS Operation}\label{sec:hris}
%============================
\noindent
In this section, we design the \gls{hris} operation with the trade-offs defined in Section~\ref{sec:intro:contributions} in mind. We first introduce the two hardware architectures, followed by the general considerations for the probe mode and two probing strategies tailored to each architecture. For clarity, we avoid overloading notation by not differentiating signals related to each architecture. Next, we outline a common reflection mode for both architectures. Finally, we discuss the computational complexity of the \gls{hris} operation for each architecture.

%----------
\subsection{The Two Hardware Architectures}\label{sec:hris:architectures}
%----------
\noindent
Figure~\ref{fig:hris-hardware} depicts the two \gls{hris} hardware architectures, {representing two implementation complexity extremes}. They follow the same basic operation from Section~\ref{sec:hris-model}, switching between operation modes by using or not the sensing hardware and loading configurations according to each mode. However, they differ in their \gls{dsp} capabilities, as follows. A \underline{\gls{pd}-enabled \gls{hris}} has a single \gls{rf}-combiner in the absorption branch, which analogically sums the signals absorbed by each element, followed by an \gls{rf}-power detector. This hardware architecture is the least complex and is limited to processing the combined received power only, that is, a \textit{single digital data stream}~\cite{albanese2021marisa,albanese2024ares}. A \underline{\gls{dsp}-enabled \gls{hris}} has an \gls{rf} chain for each element, resulting in \textit{$N$ separated digital data streams}. Thus, more advanced \gls{dsp} techniques can be applied over the $N$ acquired samples~\cite{alexandropoulos2021hybrid}.

\begin{figure}[t]
    \centering
    \includegraphics[width=\columnwidth]{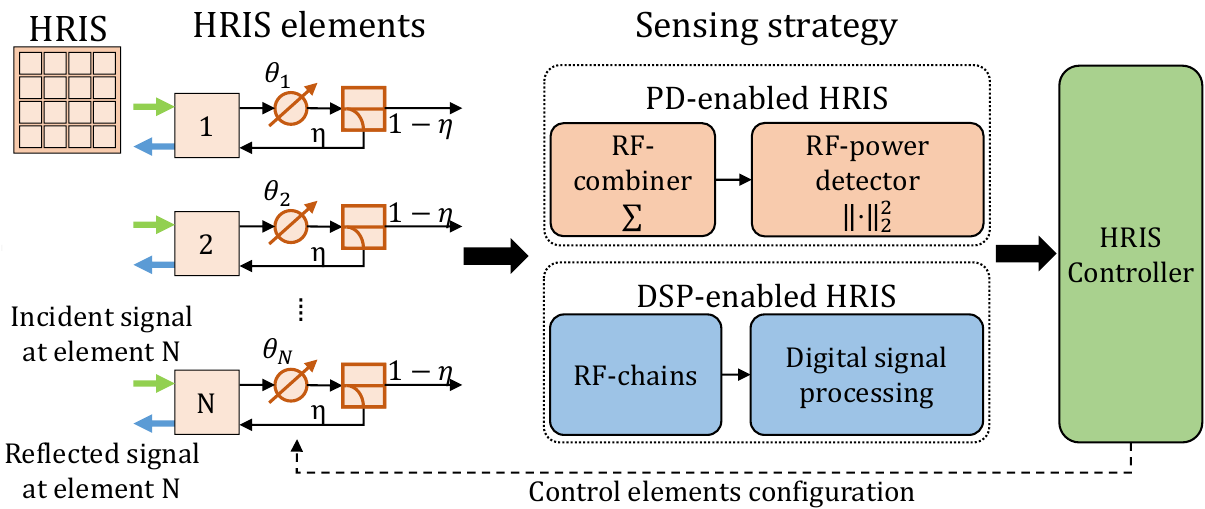}
    \vspace{-0.2in}
    \caption{
        \gls{pd}- and \gls{dsp}-enabled \gls{hris} hardware architectures.
    }
    \label{fig:hris-hardware}
\end{figure}
%

%----------
\subsection{Probe Mode: General Considerations}\label{sec:hris:probing}
%----------
\noindent
Building on Section~\ref{sec:framework}, we now discuss general considerations for the probe mode applicable to both hardware architectures. We begin with a simplifying assumption: the \gls{bs}-\gls{hris} \gls{csi} has been perfectly acquired at the \gls{hris}, \emph{e.g.}, by listening to standard synchronization and \gls{dl} pilot signals periodically transmitted by the \gls{bs}~\cite{3gpp_nr}. This assumption is supported because, following Section~\ref{sec:channel-model}, the coherence time of the \gls{bs}-\gls{hris} channel, $\mathbf{G}$, is often longer than that of the \gls{hris}-\gls{ue} channels, given the static nature of both the \gls{bs} and \gls{hris}~\cite{Yuan2022channeltracking}. Thus, we do not address the design of this aspect.

On the other hand, we address the design of detecting scheduled \glspl{ue} followed by a local \gls{chest} procedure of their \gls{csi}. From Section~\ref{sec:framework:detailed}, the \glspl{ue} transmit \gls{ul} pilots for $L$ pilot subblocks while the \gls{hris} probes during $C$ out of the $L$ subblocks. For the $c-$th subblock, the superimposed pilots, defined as $\mathbf{\Pi}[c]\in\mathbb{C}^{N \times \tau_p}$, impinging at the \gls{hris} are
\begin{equation}
    \boldsymbol{{\Pi}}[c] = \sqrt{\rho} \sum_{i\in\mathcal{K}} \mathbf{{r}}_i {\boldsymbol{\upphi}}^\transp_{p(i)},
    \label{eq:hris:probing:transmitted-pilots}
\end{equation}
where $\rho$ is the \gls{ue} transmit power, $p(i)$ denotes the pilot assigned to the $i-$th \gls{ue}, and $\mathbf{r}_i$ is the \gls{hris}-\gls{ue} channel of the $i-$th \gls{ue}, as in \eqref{eq:ricean-model}, for $c\in\mathcal{C}$. 

Per Assumption~\ref{assu:hris-configuration-change}, we establish that the \gls{hris} can change its configuration to help in probing for \glspl{ue}. For example, the \gls{hris} can ``scan'' the surrounding area by changing its configurations to detect signals coming from different directions. We refer to these as \emph{probing configurations}. Based on Assumption \ref{assu:probing-time}, the \gls{hris} can change probing configurations on a pilot subblock basis. Motivated by~\cite{albanese2021marisa,Croisfelt2023}, we will limit ourselves to the case that a \textit{probing configuration codebook} is available at the \gls{hris}, which is comprised of $C$ configurations—one configuration per pilot subblock—and is denoted as
\begin{equation}
     \Theta_{\rm P} = \left\{{\mathbf{\Theta_{\rm P}}[c]}\,|\, c\in\mathcal{C}\right\}.
     \label{eq:hris:probing:codebook}
\end{equation}
In general, the \gls{hris} must detect the scheduled \glspl{ue} and locally estimate their \gls{csi}. However, not all scheduled \glspl{ue} can be served by the \gls{hris}. For instance, a \gls{ue} located far from the \gls{hris} may be scheduled yet remain undetectable due to factors such as low received power, which may arise from an inadequate probing design or insufficient probing duration.\footnote{
    One could argue that the \gls{bs} can inform the \gls{hris} about scheduled \glspl{ue} through standard control channels. However, this does not eliminate the need for probing since a core part of it is to identify the relative positions of the \glspl{ue} at the \gls{hris}. Of course, if the \glspl{ue} are static and their channels as well (flat-fading), the \gls{hris} could probe less frequently.
} 

To design the probe mode, we employ detection and estimation theories~\cite{Kay1997detection,kay1997estimation}. In particular, we adopt two \textit{suboptimal} choices.\footnote{
    {We recall that our objective is not to optimize thoroughly the \gls{hris} operation but to comprehensively explore the fundamental trade-offs of an \gls{hris}-assisted system with a focus on robustness, as outlined in Section~\ref{sec:intro}.}
}
The first choice is \textit{channel-agnostic probing}, meaning that the probe mode does not rely on any prior knowledge of channel models, that is, the channels are treated as deterministic signals.\footnote{If partial or complete knowledge of channel models is available, more effective probe designs could be achieved.} Evidently, probing performance will vary statistically according to channel distributions and related parameters; the latter are treated as known nuisance parameters for performance evaluation~\cite{Kay1997detection}. Choosing this approach can also be justified by unfavorable characteristics of the propagation environment and application scenarios; \textit{e.g.}, rapid changes in the behaviors of \glspl{ue}~\cite{massivemimobook}. The second choice is \textit{minimal probing}, indicating that our aim is to capture essential probing functionalities without disproportionately favoring any specific hardware architecture. This is achieved by considering the simplest detection problem in the context of each hardware architecture: determining whether a signal embedded in noise is present or not~\cite{Kay1997detection}. Here, the signal of interest is treated as deterministic, while the noise distribution, although known, is not necessarily Gaussian and may have unknown parameters.

These choices can be interpreted as establishing a lower bound on probing performance for each hardware architecture, thereby offering a \textit{basic platform} for comparing the two architectures. We incentivize the study of more elaborated designs if the computational complexity remains practical. 

\paragraph*{\textbf{Output of the probe mode}} 
At the end of the \gls{hris} probing procedure, there are three main outputs that are going to be input to the reflection mode. First, the \textit{set of detected \glspl{ue}}, denoted as $\mathcal{K}_{\rm D}\!\subseteq\!\mathcal{K}$ with $|\mathcal{K}_{\rm D}|\!=\!K_{\rm D}$ and $K_{\rm D}\!\leq\!K$. Second, we note that the \textit{local \gls{csi}} needed by the \gls{hris} is the angular information of the \gls{hris}-\gls{ue} channels, $\angle\mathbf{r}_j$, which we denote as $\hat{\boldsymbol{\Theta}}_j\in\mathbb{C}^{N\times{N}}$, for $j\in\mathcal{K}_{\rm D}$. Third, a vector of \textit{relative importance weights}, denoted by $\boldsymbol{\omega}=[\omega_1,\dots,\omega_{K_{\rm D}}]^{\transp}\in\mathbb{R}_{+}^{K_{\rm D}}$, representing the relevance of each detected \gls{ue}.\footnote{
    As in~\cite{albanese2021marisa}, we focus on designing $\boldsymbol{\omega}$ by favoring \glspl{ue} according to their received amplitudes at the \gls{hris}. Other designs can be explored in the future.
}

\paragraph*{\textbf{Implementation complexity and probe distortion}} We observe that each hardware architecture leads to probe distortion with different characteristics. These differences arise from the specific configurations of the probe codebooks, $\Theta_{\rm P}$, the design of the probing scheme itself, and the influence of probing performance on reflection performance.

%----------
\subsection{PD-Enabled Probe Mode}\label{sec:hris:probing:power}
%----------

%
\subsubsection{Design of the probing configuration codebook}
Similar to~\cite{albanese2021marisa} and to overcome the lack of \gls{dsp} capabilities, the \gls{pd}-enabled \gls{hris} probes for \glspl{ue} by \textit{sweeping} through probing configurations in $\Theta_{\rm P}$. Thus, we build the probing configuration codebook to slice the 3D space into $C$ uniform sectors of interest, with $C=C_\text{el}C_\text{az}$ being decomposed into elevation and azimuth directions, respectively. The $n-$th diagonal element of the $c-$th probing configuration is 
\begin{equation}
    \left[\boldsymbol{\Theta}_{\rm P}[c]\right]_{n,n} = e^{j\langle\mathbf{k}(\mathbf{p}[{c}], \mathbf{e}), (\mathbf{e}_n - \mathbf{e})\rangle},
    \label{eq:hris:probing:power:sweeping}
\end{equation}
for $n\in\mathcal{N}$ and $c\in\mathcal{C}$, where the $c-$th probed position is
$\mathbf{p}[c]\! =\! \left[\sin\psi[c] \cos\phi[c], \sin\psi[c] \cos \phi[c], \cos\psi[c]\right]^\transp,
$
with the respective elevation and azimuth angular directions being $\psi[c] \!=\! {\pi}/{C_\text{el}} ( \text{mod}_{C_\text{el}} (c - 1) + {1}/{2} )$ and $\phi[c] \!=\! {\pi}/{C_\text{az}} ( ({c -1 - \text{mod}_{C_\text{el}} (c - 1)})/{C_\text{el}} + {1}/{2} )$.

\subsubsection{Probing procedure and performance analysis}
Consider a given pilot subblock $c$ in which the $c-$th probing configuration is loaded at the \gls{hris} according to \eqref{eq:hris:probing:power:sweeping}, for $c\in\mathcal{C}$. Let $\boldsymbol{\theta}_{\rm P}[c]\in \mathbb{C}^{N}$ denote the diagonal elements of $\mathbf{\Theta}_{\rm P}[c]$. Based on \eqref{eq:hris:probing:transmitted-pilots} and after the \gls{rf}-combiner (see Fig. \ref{fig:hris-hardware}), the received signal at the $c-$th subblock, $\mathbf{y}[c]\in\mathbb{C}^{{\tau_p}}$, is given by
\begin{equation} 
    \begin{aligned}
        (\mathbf{y}[{c}])^\transp = \sqrt{1 - \eta} \sqrt{\rho}(\boldsymbol{\theta}_{\rm P}[c])^\htransp \left( \sum_{i\in\mathcal{K}} \mathbf{{r}}_i {\boldsymbol{\upphi}}^\transp_{p(i)} \right) + (\mathbf{n}[c])^\transp,
    \end{aligned}
    \label{eq:hris:probing:power:received-signal}
\end{equation}
where $\mathbf{n}[c] = [n_1[c], \dots, n_{\tau_p}[c]]^\transp \in \mathbb{C}^{\tau_p}$ is the receiver noise at the \gls{hris} after the \gls{rf}-combiner; the noise is \gls{iid} over different subblocks and distributed as $\mathcal{CN}(\mathbf{0}, N\sigma_{\rm H}^2 \mathbf{I}_{\tau_p})$ with $\sigma^2_{\rm H}$ being the \textit{\gls{hris} noise power}. Let us focus on the $t-$th pilot and the $k-$th \gls{ue}, for $t\in\mathcal{T}_p$ and $k\in\mathcal{K}$. Based on Assumption~\ref{assu:ortho-pilots}, the above expression can be rewritten as 
\begin{equation}
    y_{t}[c]\!=\!
    \begin{cases}
        \sqrt{1-\eta}\sqrt{\rho}\sqrt{\tau_p}(\boldsymbol{\theta}_{\rm P}[c])^\htransp\mathbf{{r}}_k + {n}_{t}[c],& \text{if } p(k)\!=\!t\\
       {n}_{t}[c], & \text{o/w}.
    \end{cases}
    \label{eq:hris:probing:power:decorrelated-signal}
\end{equation} 
Let $\alpha_{t}[c]=|y_{t}[c]|^2$ denote the signal after the \gls{rf}-power detector in Fig.~\ref{fig:hris-hardware}. Then, we have that
\begin{equation}
    \alpha_{t}[c]\!=\!
    \begin{cases}
        |A_k[c]|^2 \!+ \!2 \Re\{A_k\![c] n_{t}\![c]\}\! + \!|n_{t}\![c]|^2,&\!\!\!\!\text{if }p(k)\!=\!t\\
        |{n}_{t}[c]|^2, &\!\!\!\!\text{o/w},
    \end{cases}
    \label{eq:hris:probing:power:signal-of-interest}
\end{equation}
where the amplitude $A_{k}$ is defined as $A_{k}[c]\!=\!\sqrt{1 - \eta} \sqrt{\rho} \sqrt{\tau_p} (\boldsymbol{\theta}_{\rm P}[c])^\htransp \mathbf{{r}}_k$. The \gls{pd}-enabled \gls{hris} can store and digitally process the signals $\alpha_{t}[c]$, $\forall t \in\mathcal{T}_p$, to detect the \glspl{ue}. We stress that $y_{t}[c]$ is not accessible for processing, since the \gls{pd}-enabled \gls{ris} can just measure the combined received power, $\alpha_{t}[c]$ (see Fig.~\ref{fig:hris-hardware}).

Thus, the \gls{pd}-enabled \gls{hris} detects if the $k-$th \gls{ue} is in the direction probed by the $c-$th configuration by applying the following binary hypothesis test over each pilot~\cite{Kay1997detection}:
\begin{equation}
    \begin{aligned}
        \mathcal{H}^{(k)}_0[c] &:\alpha_{t}[c] = \abs{{{n}}_{t}[c]}^2 \implies A_{k}[c] = 0, \\
        \mathcal{H}^{(k)}_1[c] &:\alpha_{t}[c] = |y_{t}[c]|^2 \implies A_{k}[c] \neq 0,
    \end{aligned}
    \label{eq:hris:probing:power:hypothesis}
\end{equation}
where the null hypothesis denotes the case in which the $k-$th \gls{ue} was not assigned to the $t-$th pilot, that is, $p(k)\neq t$ and, consequently, $A_k[c]=0$. Note that the test is performed on the amplitude $A_k[c]$, which is not directly observed from $\alpha_t[c]$, as seen in \eqref{eq:hris:probing:power:signal-of-interest}. Hence, we need to estimate $A_k[c]$ from $\alpha_t[c]$. 
Let $f_\mathrm{MLE}$ denote the \gls{mle}. The \gls{mle} for $|A_k[c]|^{2}$ from $\alpha_t[c]$ is $f_\mathrm{MLE}(A_{k}[c])\! =\! \alpha_t[c]$. Thus, the \gls{pd}-enabled \gls{hris} decides $\mathcal{H}^{(k)}_1[c]$ if~\cite[p. 200]{Kay1997detection}:
\begin{equation}
    \frac{p(\alpha_{t}[c];{f_\mathrm{MLE}(A_{k}[c])},\mathcal{H}^{(k)}_1[c])}{p(\alpha_{t}[c];{\mathcal{H}^{(k)}_0[c])}}>\epsilon_s,
    \label{eq:hris:probing:power:general-glrt}
\end{equation}
where $\epsilon_s$ is a threshold parameter. The detailed test description can be seen in Appendix~\ref{appx:proof-power}, which relies on an approximation based on ignoring the cross-term in~\eqref{eq:hris:probing:power:signal-of-interest}. We evaluate the approximated performance of the \gls{pd}-enabled \gls{hris} below.
\begin{corollary}[PD-enabled probing performance]
    An approximated closed-form expression of the performance of the \gls{pd}-based probe mode given by the test in~\eqref{eq:hris:probing:power:hypothesis} can be found in the asymptotic case of $N\xrightarrow{}\infty$ as~\cite{Kay1997detection}:
    \begin{equation}
        P^{(k)}_\mathrm{D}[c] = e^{-\frac{1}{2N\sigma^2_\mathrm{H}}(\epsilon^\prime_s-\alpha_{t}[c])} \text{ and }
        P^{(k)}_\mathrm{FA}[c] = e^{-\frac{1}{2N\sigma^2_\mathrm{H}}\epsilon^\prime_s},
        \label{eq:hris:power:performance}
    \end{equation}
    where $P^{(k)}_{\rm D}[c]$ and $P^{(k)}_\mathrm{FA}[c]$ are the probabilities of detection and false alarm for detecting the $k-$th \gls{ue} in the $c-$th pilot subblock, respectively, for $c\in\mathcal{C}$ and $k\in\mathcal{K}$. The threshold parameter $\epsilon^\prime_s$ is proportional to $\epsilon_s$ in \eqref{eq:hris:probing:power:general-glrt}.
    \label{corollary:power:performance}    
\end{corollary}
\begin{proof}
    The proof is given in Appendix~\ref{appx:proof-power}.
\end{proof}

We note that the above performance measures overestimate the real performance due to approximations made. Moreover, the measures are stochastic and vary on a coherence-block basis with factors such as the positions of \glspl{ue}.

\subsubsection{Output} 
After performing the test in \eqref{eq:hris:probing:power:hypothesis} over all $\tau_p$ pilots, the \gls{hris} stores the detected \glspl{ue} in the set $\mathcal{K}_{\rm D}[c]=\{k \in \mathcal{K}\,|\, \mathcal{H}^{(k)}_1[c] \text{ is true}\}$. This is repeated and aggregated over all $C$ pilot subblocks. The \gls{hris} then stores all the detected \glspl{ue} along with their corresponding probing configurations that achieved the highest received power as:
\begin{equation}
    \mathcal{K}_{\rm D}=\bigcup_{c\in\mathcal{C}}\mathcal{K}_{\rm D}[c] \,\, \text{and} \,\, c_j = \argmax_{c\in\mathcal{C}} \alpha_{p(j)}[c], \, \forall j \in \mathcal{K}_{\rm D}.  
    \label{eq:power:end}
\end{equation}
The \gls{pd}-enabled \gls{hris} cannot explicitly estimate the \gls{hris}-\gls{ue} channels, $\mathbf{r}_j$, of the detected \glspl{ue} since it is limited to observe signal power, as in \eqref{eq:hris:probing:power:signal-of-interest}. Therefore, the best this \gls{hris} can do is to use the probing configuration that achieved the highest received power as an estimate of the local \gls{csi} for each detected \gls{ue}. Thus, the local \gls{csi} at the \gls{pd}-enabled \gls{hris} is 
\begin{equation}
    \hat{\boldsymbol{\Theta}}_j = \boldsymbol{\Theta}_{\rm P}[c_j],\,\forall j\in\mathcal{K}_{\rm D},
    \label{eq:hris:probing:power:output}
\end{equation}
where $c_j$ comes from~\eqref{eq:power:end} and relative importance weights are
\begin{equation}
    {\omega}_j=
    {
        \sqrt{\alpha_{p(j)}[c_j]}
    }
    /
    {
    \sum_{j^{\prime} \in \mathcal{K}_{\rm D}} \sqrt{\alpha_{p(j^{\prime})}[c_{j^{\prime}}]}
    }
    ,\,\forall j\in\mathcal{K}_{\rm D}.
    \label{eq:hris:probing:power:weights}
\end{equation}
For the \gls{pd}-enabled \gls{hris}, testing on a subblock basis is crucial, as the local \gls{csi} estimation relies on this structure.

%----------
\subsection{DSP-Enabled Probe Mode}\label{sec:hris:probing:signal}
%----------

%
\subsubsection{Design of the probing configuration codebook}
A \gls{dsp}-enabled \gls{hris} can process the received signals coming from all elements simultaneously, and, thus, it can always reverse back the effect of any impressed probing configuration $\mathbf{\Theta}_{\rm P}[c]$ digitally at the price of increased computational effort. This involves multiplying a signal received at a given sample of the $c-$th pilot subblock by $\mathbf{\Theta}^{-1}_{\rm P}[c]$. Thus, we assume that the probing configuration codebook $\Theta_{\rm P}$ is $\mathbf{\Theta}_{\rm P}[c] = \mathbf{I}_N$, $\forall c \in\mathcal{C}$.

\subsubsection{Probing procedure and performance analysis}
Based on~\eqref{eq:hris:probing:transmitted-pilots} and the above $\Theta_{\rm P}$, the received signal at the $c-$th pilot subblock, $\mathbf{Y}_{c}\in\mathbb{C}^{N \times \tau_p}$, is given by
\begin{equation} 
    \mathbf{Y}[c] = \sqrt{1 - \eta} \sqrt{\rho} \sum_{i\in\mathcal{K}} \mathbf{{r}}_i {\boldsymbol{\upphi}}^{\transp}_{p(i)} + \mathbf{N}[c],
    \label{eq:hris:probing:signal:received-signal}
\end{equation}
where $\mathbf{N}[c]\in\mathbb{C}^{N \times \tau_p}$ is the receiver noise matrix with columns distributed according to $\mathbf{n}_{t}[c]\! \sim\! \mathcal{CN}(\mathbf{0}, \sigma_{\mathrm{H}}^2 \mathbf{I}_N)$ with noise \gls{iid} over subblocks. Unlike the \gls{pd}-enabled, the \gls{dsp}-enabled \gls{hris} can process the received signal over the element-dimension, $N$, the pilot-dimension, $\tau_p$, and the pilot-subblock-dimension, $c$. We explore this next. Let us focus on the $t-$th pilot and the $k-$th \gls{ue}, for $t\in\mathcal{T}_p$ and $k\in\mathcal{K}$. We start by processing over the pilot dimension. The \gls{hris} de-correlates the received signal \gls{wrt} the $t-$th pilot as: 
\begin{equation} 
    \tilde{\mathbf{y}}_{t}[c]\!=\!\mathbf{Y}[c]{\boldsymbol{\upphi}}^{*}_t\!=\! 
    \begin{cases}
        \sqrt{1-\eta} \sqrt{\rho} \tau_p \mathbf{{r}}_k + \tilde{\mathbf{n}}_{t}[c],& \hspace{-2mm}\text{if }p(k)=t\\
        \tilde{\mathbf{n}}_{t}[c],& \hspace{-2mm}\text{o/w},
    \end{cases}
    \label{eq:hris:signal:decorrelate}
\end{equation}
where $\tilde{\mathbf{n}}_t[c]\!\sim\! \mathcal{CN}(\mathbf{0}, \tau_p\sigma_{\mathrm{H}}^2 \mathbf{I}_N)$. Next, to combat noise, the above signals can be averaged over subblocks as
\begin{equation}
    \check{\mathbf{y}}_{t}=\frac{1}{C}\sum_{c\in\mathcal{C}}\tilde{\mathbf{y}}_{t}[c]=
    \begin{cases}
        \sqrt{1 - \eta} \sqrt{\rho} \tau_p \mathbf{{r}}_k + \check{\mathbf{n}}_{t},& \text{if } p(k)=t\\
        \check{\mathbf{n}}_{t}[c], & \text{o/w},
    \end{cases}
    \label{eq:hris:probing:signal:averaged-subblocks}
\end{equation}
where $\check{\mathbf{n}}_{p(k)}\!\sim\!\mathcal{CN}(0,{\tau_p\sigma_{\mathrm{H}}^2}/{C}) \mathbf{I}_N)$. Let $\mathbf{si}_t\!=\!\sqrt{1 - \eta} \sqrt{\rho} \tau_p \mathbf{{r}}_k$ be the complex signal observed if the $k-$th \gls{ue} transmitted the $t-$th pilot, that is, $p(k)=t$ and $\mathbf{si}_t\!=\!\mathbf{0}$ otherwise. The \gls{dsp}-enabled \gls{hris} can store and digitally process the signals $\check{\mathbf{y}}_{t}$, $\forall t\in\mathcal{T}_p$, to detect the \glspl{ue}.

Thus, the \gls{dsp}-enabled \gls{hris} detects the $k-$th \gls{ue} by applying the following binary hypothesis test over each pilot~\cite{Kay1997detection}:
\begin{equation}
    \begin{aligned}
        \mathcal{H}^{(k)}_0 &: \check{\mathbf{y}}_{t} = \check{\mathbf{n}}_{t} & \implies \mathbf{si}_{t} = \mathbf{0},\\
        \mathcal{H}^{(k)}_1 &: \check{\mathbf{y}}_{t} = \mathbf{si}_{t} + \check{\mathbf{n}}_{t} & \implies \mathbf{si}_{t} \neq \mathbf{0},
    \end{aligned}
    \label{eq:hris:signal:hypothesis}
\end{equation}
where, as before, the null hypothesis denotes the case in which the $k-$th \gls{ue} was not assigned to the $t-$th pilot. Unlike the \gls{pd}-enabled, the detection is now independent on the pilot subblocks due to the higher \gls{dsp} capability. Thus, the \gls{dsp}-enabled \gls{hris} decides $\mathcal{H}^{(k)}_1$ if~\cite[p.~500]{Kay1997detection}:
\begin{equation}
    \dfrac{2C}{\tau_p \sigma^2_{\rm H}}\ltwonorm{\check{\mathbf{y}}_t}^2 > \epsilon_s,
    \label{eq:hris:probing:detection:signal:test}
\end{equation}
where $\epsilon_s$ is a threshold parameter.
We evaluate the performance of the \gls{dsp}-enabled \gls{hris} probe mode below.

\begin{corollary}[DSP-enabled probing performance]
    A closed-form expression of the performance of the {\gls{dsp}-enabled} \gls{hris} probe mode given by the test in~\eqref{eq:hris:probing:detection:signal:test} can be found as~\cite{Kay1997detection} 
    \begin{equation}
        P^{(k)}_\text{D} = Q_{\chi^{2}_{2N}(\mu)}(\epsilon_s) \text{ and } P^{(k)}_\text{FA} = Q_{\chi^{2}_{2N}}(\epsilon_s),
    \label{eq:hris:signal:performance}
    \end{equation}
    where $P^{(k)}_{\rm D}$ and $P^{(k)}_\mathrm{FA}$ are the probabilities of detection and false alarm for detecting the $k-$th \gls{ue}, respectively, and 
    $\mu={2(1-\eta)\rho\tau_p}\ltwonorm{\mathbf{r}_{k}}^2/{\sigma^{2}_{\rm H}}, \text{ for } k\in\mathcal{K}.$
    \label{corollary:signal:performance}    
\end{corollary}
\begin{proof}
    The proof follows directly from~\cite[p.~500]{Kay1997detection}.
\end{proof}
Unlike Corollary~\ref{corollary:power:performance}, the probing performance is now independent of the pilot subblocks. However, it remains stochastic and exhibits variability on a coherence-block basis, influenced by factors such as the positions of the \glspl{ue}. 

\subsubsection{Output} 
After performing the test in \eqref{eq:hris:probing:detection:signal:test} over all $\tau_p$ pilots, the \gls{hris} stores the detected \glspl{ue} in the set $\mathcal{K}_{\rm D}\!=\!\{t | \mathcal{H}^{(t)}_1 \text{ is true},\, t\! \in\! \mathcal{T}_{p}\}$. Unlike the \gls{pd}-enabled, the \gls{dsp}-enabled \gls{hris} can explicitly estimate the angular information of the \gls{hris}-\gls{ue} channels, $\angle\mathbf{r}_j$, for $j\in\mathcal{K}_{\rm D}$, by exploiting the signal in~\eqref{eq:hris:probing:signal:averaged-subblocks} to perform such estimation. We now describe such an estimation process for the $j-$th detect \gls{ue} assigned to the $t-$th pilot with $p(j)\!=\!t$, for $j\!\in\!\mathcal{K}_{\rm D}$.
Let ${\boldsymbol{\theta}}_{j}\!=\!\angle\mathbf{r}_j\!\in\!\mathbb{C}^{N}$ denote the angular information of the \gls{hris}-\gls{ue} channel, which we are interested in estimating. From~\eqref{eq:hris:probing:signal:averaged-subblocks}, we observe that the angular information contained in $\mathbf{si}_{j}$ is equivalent to the one contained in $\mathbf{r}_j$, that is, $\angle\mathbf{si}_{j}\equiv\angle\mathbf{r}_j$, since $\mathbf{si}_{j}$ is proportional to $\mathbf{r}_j$. The estimation of ${\boldsymbol{\theta}}_j$ is then based on rewriting the signal in \eqref{eq:hris:probing:signal:averaged-subblocks} as $\check{\mathbf{y}}_{t}\!=\!\mathbf{si}_{j} + \check{\mathbf{n}}_{t}.$ Thus, the \gls{hris} estimates ${\boldsymbol{\theta}}_j$ as
\begin{equation}
    \hat{\boldsymbol{\theta}}_j = \exp \left(1j \arctan\left(\frac{\mathfrak{I}(\check{\mathbf{y}}_{t})}{\mathfrak{R}(\check{\mathbf{y}}_{t})}\right) \right),
    \label{eq:hris:probing:signal:estimation}
\end{equation}
where the $\exp(\cdot)$ and $\arctan(\cdot)$ functions are applied element-wise over the vector entries, and we use the notation $1j$ to stress the difference between the \gls{ue} index and the imaginary unit. The estimation error can be numerically approximated as the variance of the signal $\check{\mathbf{y}}_{t}\!\sim\!\mathcal{CN}(\mathbf{si}_j,({\tau_p\sigma_{\mathrm{H}}^2}/{C}) \mathbf{I}_N)$. Thus, the local \gls{csi} at the \gls{dsp}-enabled \gls{hris} is 
\begin{equation}
    \hat{\boldsymbol{\Theta}}_j = \diag(\hat{\boldsymbol{\theta}}_j),\, \forall j\in\mathcal{K}_{\rm D}.
    \label{eq:hris:probing:signal:output}
\end{equation}
Similar to before, we compute the corresponding weights as
\begin{equation}
    {\omega}_j= 
    {
        \ltwonorm{\check{\mathbf{y}}_{p(j)}}
    }
    /
    {
        \sum_{j^{\prime}\in\mathcal{K}_{\rm D}}\ltwonorm{\check{\mathbf{y}}_{p(j^{\prime})}}
    },\, \forall j\in\mathcal{K}_{\rm D}.
    \label{eq:hris:probing:signal:weights}
\end{equation}

%----------
\subsection{Reflection Mode}\label{subsec:hris:refl-mode}
%----------
\noindent
We design the reflection mode based on the outputs of the probe mode, specifically, the set of detected \glspl{ue}, $\mathcal{K}_D$, the local \gls{csi}, $\hat{{\boldsymbol{\Theta}}}_{j}$, and the relative importance weights, $\omega_j$, $\forall j\in\mathcal{K}_D$. The latter two quantities are provided in eqs.~\eqref{eq:hris:probing:power:output} and~\eqref{eq:hris:probing:power:weights} for a \gls{pd}-enabled \gls{hris}, and in eqs.~\eqref{eq:hris:probing:signal:output} and \eqref{eq:hris:probing:signal:weights} for a \gls{dsp}-enabled while the former is the collective result of the tests in \eqref{eq:hris:probing:power:hypothesis} and \eqref{eq:hris:signal:hypothesis}, respectively. Thus, the design of the reflection mode remains independent of the \gls{hris} hardware architecture, although its performance eventually differs due to distinct probe performance. In principle, the \gls{hris} would load \textit{one} reflection configuration to assist the \gls{ul} data traffic and \emph{another} for \gls{dl}. However, by leveraging channel reciprocity, we note that the reflection configuration for the \gls{dl} is the complex conjugate of the one used during \gls{ul}. Hence, we focus solely on designing a single reflection configuration for \gls{ul}, denoted as $\hat{\boldsymbol{\Theta}}_{\rm R}$. {In particular, we adopt the reflection design from~\cite{albanese2021marisa}, whose goal is to maximize the received \gls{snr} of the detected \glspl{ue} while taking into account their relative importance weights. This is obtained by setting $\hat{\boldsymbol{\Theta}}_{\rm R}$ as}
\begin{equation}
    \hat{\boldsymbol{\Theta}}_{\rm R} = \mathbf{\Theta}_{\rm B} \circ
    \sum_{k\in\mathcal{K}_{\rm D}} \omega_k \hat{\mathbf{\Theta}}_{k}^{*},
    \label{eq:hris:reflection:configuration}
\end{equation}
where $\boldsymbol{\Theta}_{\rm B}=\diag(\mathbf{a}_{\rm H}(\mathbf{b}))$ denotes the perfect \gls{csi} of the \gls{hris}-\gls{bs} channel, $\mathbf{G}$. {Similar to probing, we argue that this is a minimal reflecting design (see details in~\cite{albanese2021marisa}).}

\paragraph*{{\textbf{Evaluating reflecting performance}}}
We now present a metric to evaluate the designed reflection configuration. In an ideal scenario of perfect probing, the \gls{hris} would employ the following optimal reflection configuration, assuming that all \glspl{ue} are detected and their \gls{csi} is accurately estimated:
\begin{equation}
    {\boldsymbol{\Theta}}^{\star}_{\rm R} = \mathbf{\Theta}_{\rm B} \circ \sum_{k\in\mathcal{K}} \omega^{\star}_k \diag(\boldsymbol{\theta}^{\star}_k)^{*},
\end{equation}
where we use $(\cdot)^{\star}$ to denote optimal in the sense {defined in~\cite{albanese2021marisa}} with 
$
    \boldsymbol{\theta}^{\star}_k\!=\! \exp({j \arctan\left({\mathfrak{I}({\mathbf{r}}_k)}/{\mathfrak{R}({\mathbf{r}}_k)}\right)}) 
$
and 
$
\omega^{\star}_k\! =\! {\lVert\mathbf{r}_k\rVert_2} /{\sum_{i\in\mathcal{K}}\lVert\mathbf{r}_i\rVert_2}.
$ 
Thus, a metric for evaluating reflection accuracy is the \gls{nmse}:
\begin{equation}
    \mathrm{NMSE}_{\rm H}(\varpi) = {\lVert\hat{\boldsymbol{\theta}}_{\rm R}-{\boldsymbol{\theta}}^{\star}_{\rm R}\rVert^2_2} /{\lVert{\boldsymbol{\theta}}^{\star}_{\rm R}\rVert^2_2},
    \label{eq:hris:reflection:nmse}
\end{equation}
which is a function of the relative probe duration $\varpi$ (see~Def.~\ref{definition:varpi}) and where $\hat{\boldsymbol{\theta}}_{\rm R}$ and ${\boldsymbol{\theta}}^{\star}_{\rm R}$ are the respective diagonals of $\hat{\boldsymbol{\Theta}}_{\rm R}$ and ${\boldsymbol{\Theta}}^{\star}_{\rm R}$. Observe that $\mathrm{NMSE}_{\rm H}(\varpi)$ captures implementation complexity, as it depends on the chosen \gls{hris} hardware architecture, and probe distortion, as it depends on probing performance; it also statistically varies over coherence blocks with factors such as the positions of \glspl{ue}, and channel and noise realizations.

%%%%%
\subsection{Complexity Analysis: HRIS Operation Modes}\label{sec:hris:complexity}
%%%%%
\noindent
For the \gls{pd}-enabled \gls{hris}, the \gls{rf}- combiner and power detector can be implemented with analog circuitry. Thus, computing is required for~\eqref{eq:hris:probing:power:general-glrt},~\eqref{eq:power:end}, and~\eqref{eq:hris:probing:power:weights}, yielding in a total of $C + K_{\max} C + 3 K_{\max}$ element-wise operations. The \gls{dsp}-enabled \gls{hris} performs~\eqref{eq:hris:signal:decorrelate},~\eqref{eq:hris:probing:signal:averaged-subblocks},~\eqref{eq:hris:probing:detection:signal:test},~\eqref{eq:hris:probing:signal:estimation}, and~\eqref{eq:hris:probing:signal:weights}, resulting in a total of $NK_{\max}^3 + 5NK_{\max} + CN + 2N$ element-wise operations. By adding up $2 N K_{\max}$ operations to compute the diagonal reflection configuration in~\eqref{eq:hris:reflection:configuration}, we obtain the computational complexities of $\mathcal{O}(K_{\max} (2N + C + 3) + C)$ for the \gls{pd}-enabled \gls{hris} and of $\mathcal{O}(N (K_{\max}^3 + 7 K_{\max} + C + 2))$ for the \gls{dsp}-enabled one. Hence, observe that the complexity of the \gls{pd}-enabled \gls{hris} increases linearly with system parameters while the complexity of the \gls{dsp}-enabled scales cubically with $K_{\max}$; more concerning is the comparison between $2NK_{\max}$ for the former against $NK_{\max}^3$ for the latter.

%============================
\section{Designing the mMIMO Operation}\label{sec:mmimo}
%============================
\noindent
In this section, we adapt the design of a traditional \gls{mmimo} system~\cite{massivemimobook} to account for the influence of \gls{hris} operation, {having the trade-offs defined in Section~\ref{sec:intro:contributions} in mind.} For generality, we will assume a generic \gls{hris} hardware architecture, interchangeable with either \gls{pd}-enabled, \gls{dsp}-enabled, or other architectures. For simplicity, we assume the \gls{comm} phase includes only \gls{ul} traffic, with $\tau_d\!=\!0$ and $\tau_u\!=\!\tau_c-L\tau_p$; extension to the \gls{dl} case is straightforward. 

%----------
\subsection{CHEST Phase}
%----------
\noindent
Based on Sections~\ref{sec:system-model} and~\ref{sec:framework}, we can now formally define the following two equivalent channels for the $i-$th \gls{ue}:
\begin{equation}
        {\mathbf{h}}_{\rm{P}, i}[l]\!=\!\mathbf{h}_{\mathrm{DR},i}+\sqrt{\eta}\mathbf{G}\mathbf{\Theta}_{\rm P}[l]\mathbf{{r}}_i 
        \text{ and }
        {\mathbf{h}}_{\rm{R}, i}\!=\!\mathbf{h}_{\mathrm{DR},i}+\sqrt{\eta}\mathbf{G}\hat{\mathbf{\Theta}}_{\rm R}\mathbf{{r}}_i,\!\!
    \label{eq:mmimo:chest:equivalent-channels}
\end{equation}
where ${\mathbf{h}}_{\rm{P},i}[l]$ denotes the \textit{probing equivalent channel} during the $l-$th pilot subblock and ${\mathbf{h}}_{\rm{R},i}$ represents the \textit{reflecting equivalent channel}, for $l\in\mathcal{C}$ and $i\in\mathcal{K}$. Recall that $\mathcal{C}$ is the set of pilot subblocks in which the \gls{hris} probes (see Def.~\ref{assu:probing-time}), $\mathbf{\Theta}_{\rm P}[l]$ is the probing configuration of the $l-$th subblock, as in~\eqref{eq:hris:probing:codebook}, and $\hat{\mathbf{\Theta}}_{\rm R}$ is the reflection configuration, as in~\eqref{eq:hris:reflection:configuration}. Following~\eqref{eq:hris:probing:transmitted-pilots} and the above definitions, the \gls{bs} receives the superimposed pilots at the $l-$th subblock, $\mathbf{Z}[l]\in\mathbb{C}^{M\times{\tau_p}}$, as:
\begin{equation}
    \mathbf{Z}[l]=
    \sqrt{\rho}
    \begin{cases}
        \sum_{i\in\mathcal{K}}{\mathbf{h}}_{\rm{P},i}[l] {\boldsymbol{\upphi}}^{\transp}_{p(i)}+\mathbf{W}[l], & \text{if } l\in\mathcal{C}\\
        \sum_{i\in\mathcal{K}}{\mathbf{h}}_{\rm{R},i} {\boldsymbol{\upphi}}^{\transp}_{p(i)}+\mathbf{W}[l], & \text{o/w},
    \end{cases}
    \label{eq:mmimo:received-pilots}
\end{equation}
where $\mathbf{W}[l]\in\mathbb{C}^{M\times{\tau_p}}$ is the \gls{bs} receiver noise whose \gls{iid} entries follow $\mathcal{CN}(0,\sigma_{\rm B}^2)$ with $\sigma^2_{\rm B}$ being the \textit{\gls{bs} noise power}; the noise is also \gls{iid} over subblocks. In principle, the oblivious \gls{bs} aims to estimate the stable reflecting equivalent channels $\{{\mathbf{h}}_{\mathrm{R},i}\}_{i\in\mathcal{K}}$ from the collected $\mathbf{Z}[l]$, for $l\in\mathcal{C}$. However, this estimation process suffers from the probe distortion, which is now formally characterized by the summation of probing equivalent channels, $\sum_{i\in\mathcal{K}}{\mathbf{h}}_{\rm{P},i}[l]$. 

{Under the assumption of an oblivious \gls{bs}, the \gls{bs} carries out the following \gls{chest} procedure.} For the sake of argument, we focus on a single \gls{ue} $k$ that was assigned the $t-$th pilot, $p(k)\!=\!t$, for $t\!\in\!\mathcal{T}_p$ and $k\!\in\!\mathcal{K}$. Let $\mathbf{Z}\!=\![\mathbf{Z}[1],\dots,\mathbf{Z}[L]]\!\in\!\mathbb{C}^{M\times{L\tau_p}}$ be the horizontally concatenated matrix of all pilot subblocks received by the \gls{bs}. Denote as $\boldsymbol{\upphi}_{Lt}\!=\![{\boldsymbol{\upphi}}_{t};\dots;{\boldsymbol{\upphi}}_{t}]\!\in\!\mathbb{C}^{L\tau_p}$ the vector containing the $t-$th pilot repeated $L$ times. The \gls{bs} first takes the mean of the de-correlated received signals in \eqref{eq:mmimo:received-pilots}, yielding in $\bar{\mathbf{z}}_k\!=\!\frac{1}{L}\mathbf{Z}{\boldsymbol{\upphi}}_{Lt}^{*}$ as 
\begin{equation}
        \bar{\mathbf{z}}_k\stackrel{(a)}{=} \sqrt{\rho} \tau_p \underbrace{\left(
        \frac{1}{L}\sum_{l=1}^{C}\mathbf{{h}}_{\mathrm{P}, k}[l]+(1-\varpi){\mathbf{h}}_{\mathrm{R}, k}\right)}_{=\bar{\mathbf{h}}_k} + \frac{1}{L} \sum_{l=1}^{L} {\mathbf{w}_t}[l],
        \label{eq:mmimo:eq-channel}     
\end{equation}
where ${\mathbf{w}_t}[l]\!\sim\!\mathcal{CN}(\mathbf{0},\tau_p \sigma_{\rm B}^2 \mathbf{I}_{M})$ is the equivalent receiver noise and $\bar{\mathbf{h}}_k\in\mathbb{C}^{M}$ is defined as the \textit{average equivalent channel}. In $(a)$, we have used Def.~\ref{definition:varpi} for $\varpi$ as the relative probe duration. Below, we provide the \gls{ls} estimate of $\bar{\mathbf{h}}_k$.\footnote{
    To align with the {channel-agnostic probing design}, we assume that the \gls{bs} has no prior knowledge of the channel statistics. Otherwise, Bayesian estimation methods~\cite{kay1997estimation} could be employed to enhance performance further.
}

\begin{corollary}[CHEST at the BS]\label{corollary:ls-channel-estimation}
    The \gls{ls} estimate of the average equivalent channel $\bar{\mathbf{h}}_k$ based on $\bar{\mathbf{z}}_{k}$ is
    $
        \hat{\mathbf{h}}_k\!=\!{(\sqrt{\rho}\tau_p)^{-1}}\bar{\mathbf{z}}_{k}\text{ with }\hat{\mathbf{h}}_k\!\sim\!\mathcal{CN}(\bar{\mathbf{h}}_k,\hat{\sigma}^{2}\mathbf{I}_M),
    $
    where $\hat{\sigma}^2\!=\!{\sigma^2_{\rm B}}/{(L\rho\tau_p)}$ denotes the variance of the estimate.
\end{corollary}
\begin{proof}
    The proof follows~\cite[p. 225]{kay1997estimation}.
\end{proof}

\paragraph*{\textbf{Measuring probe distortion}}
Thus, instead of estimating $\mathbf{h}_{\mathrm{R},i}$, the \gls{bs} estimated $\bar{\mathbf{h}}_k$. To evaluate the quality of this \gls{csi} and capture the impact of probe distortion, we define the following average \gls{nmse}:
\begin{equation}
    \overline{\mathrm{NMSE}}_{\rm{B}, k}(\varpi)\!=\!
    \E{
        \frac{
            \ltwonorm{
                \hat{\mathbf{h}}_k\! -\! {\mathbf{h}}_{\mathrm{R},k}
            }^2
            }{
            \ltwonorm{
                {\mathbf{h}}_{\mathrm{R},k}
            }^2
            } 
    } \!=\! 
    \frac{
        M\hat{\sigma}^2\!+\!\ltwonorm{\bar{\mathbf{h}}_k\!-\!{\mathbf{h}}_{\mathrm{R},k}}^2
    }{
        \ltwonorm{{\mathbf{h}}_{\mathrm{R},k}}^{2}
    },
    \label{eq:mmimo:chest:nmse1}
\end{equation}
which is a function of the relative probe duration $\varpi$ (see~Def.~\ref{definition:varpi}) and where the expectation was taken over noise realizations. Per~\eqref{eq:mmimo:chest:equivalent-channels} and Corollary~\ref{corollary:ls-channel-estimation}, we rewrite~\eqref{eq:mmimo:chest:nmse1} as
\begin{equation}
    \overline{\mathrm{NMSE}}_{\rm{B}, k}(\varpi)\!=\!
    \frac{
        \frac{M}{L}\frac{\sigma_{\rm B}^2}{\rho\tau_p}\!+\!
        \frac{\eta}{L^2}\ltwonorm{\mathbf{G}\left(\sum_{l=1}^{C}\boldsymbol{\Theta}_{\rm P}[l]
        \!-\!
        C\boldsymbol{\hat{\Theta}}_{\mathrm{R}}\right)\mathbf{{r}}_k}^2
    }{
        \frac{\eta}{L^2}\ltwonorm{\mathbf{G}\boldsymbol{\hat{\Theta}}_{\mathrm{R}}\mathbf{{r}}_k}^2
    },
    \label{eq:mmimo:chest:nmse2}
\end{equation}
where the first left-hand side term in the sum of the numerator accounts for the \emph{true \gls{ls} estimation error}, that is, if ${\mathbf{h}}_{\mathrm{R},k}$ was to be estimated without probe distortion, while the second term evaluates the effect of probe distortion. Observe that $\overline{\mathrm{NMSE}}_{\rm{B},k}$ also captures implementation complexity since it depends on the \gls{hris} hardware architectures; it also statistically varies over coherence blocks with factors such as the positions of \glspl{ue} and channel realizations.

\begin{remark}
    From eq.~\eqref{eq:mmimo:chest:nmse2}, we can draw two main conclusions. First, there would be no probe distortion if: a) (obliviously) the probe mode is not employed, \emph{i.e.}, $C=0$ or $\varpi=0$ or b) the probing configurations were identical to the reflection configuration, $\boldsymbol{\Theta}_{\rm P}[l]\!=\!\hat{\boldsymbol{\Theta}}_{\mathrm{R}}, \forall l\in\mathcal{C}$. However, obtaining $\hat{\boldsymbol{\Theta}}_{\mathrm{R}}$ before initiating the probing mode is infeasible, as it represents the primary objective of the probing process itself.\footnote{
        Note that this reveals one potential approach to mitigate probe distortion: implement intelligent probing strategies using outdated location information as a guide, where again if the \glspl{ue} were static and their channels as well (flat-fading), the \gls{hris} could probe less frequently. 
    } Second, the estimation error is maximized when the probe mode occupies the entire \gls{chest} phase, \emph{i.e.}, $C=1$ and $\varpi=1$. From this discussion, an alternative way to measure the probe distortion is
    $
    % \begin{equation} 
        {
            \left\lVert\frac{1}{C}\sum_{l=1}^{C}\boldsymbol{\Theta}_{\rm P}[l]-\hat{\boldsymbol{\Theta}}_{\mathrm{R}}\right\rVert^{2}_{F}
        }
        /
        {
            \lVert\hat{\boldsymbol{\Theta}}_{\mathrm{R}}\rVert^{2}_{F}
        },
        \label{eq:chest:distortion}
    %\end{equation}
    $
    which measures how different the probing and reflection configurations are on average using a Frobenius norm. However, this metric does not capture the real impact of probe distortion on communication performance, which is addressed next.
    \label{remark:distortion}
\end{remark}
%

%----------
\subsection{COMM Phase}
%----------
\noindent
During the \gls{comm} phase, the oblivious \gls{bs} exploits the probe-distorted \gls{csi} in Corollary \ref{corollary:ls-channel-estimation} to spatially separate the \glspl{ue} while the \gls{hris} is in the reflection mode. Let $\mathbf{v}_k$ denote the receive combining vector for the $k-$th \gls{ue}, which is a function of the probe-distorted \gls{csi}, $\hat{\mathbf{h}}_k$, for $k\in\mathcal{K}$. Here, we focus on the specific case of the \gls{mr} scheme with $\mathbf{v}_k\!=\!\hat{\mathbf{h}}_k$~\cite{massivemimobook}, as, due to space limitations, alternative choices cannot be thoroughly addressed. By focusing on a particular sample of the $\tau_u$ samples, the \gls{bs} estimates a payload signal sent by the $k-$th \gls{ue} as follows~\cite{massivemimobook}:
\begin{equation}
    \hat{s}_k=\mathbf{v}_k^{\htransp}{\mathbf{h}}_{\mathrm{R},k}s_k +\sum_{\substack{i\in\mathcal{K}, i\neq k}}\mathbf{v}_k^{\htransp}{\mathbf{h}}_{\mathrm{R},i}s_i+\mathbf{v}_k^{\htransp}\mathbf{o},
    \label{eq:estimated-signal}
\end{equation}
where ${\mathbf{h}}_{\mathrm{R},k}$ is defined in \eqref{eq:mmimo:chest:equivalent-channels}, $s_j\sim\mathcal{CN}(0,\rho)$ is a random data signal for $j-$th \gls{ue} with $j\in\mathcal{K}$, and $\mathbf{o}\sim\mathcal{CN}(\mathbf{0},\sigma^2_{\rm B}\mathbf{I}_M)$ is the \gls{bs} receiver noise. 

{We proceed by discussing the impact of the \gls{hris} operation on communication performance more formally, motivated by the autonomous-RIS trade-off introduced in Section~\ref{sec:intro:contributions}.} Let $\bar{\mathbf{h}}_{{\rm P},k}=\frac{1}{L}\sum_{l=1}^{C}\mathbf{{h}}_{\mathrm{P}, k}[l]$. From Corollary~\ref{corollary:ls-channel-estimation}, we rewrite the probe-distorted \gls{csi} estimated at the \gls{bs} as
\begin{equation}
    \hat{\mathbf{h}}_k\sim\mathcal{CN}\left(\bar{\mathbf{h}}_{{\rm P},k}+(1-\varpi)\mathbf{h}_{\mathrm{R}, k},\hat{\sigma}^{2}\mathbf{I}_M\right) \text{ for } k\in\mathcal{K}.
\end{equation}
From this, we can see that the effect of the probe distortion is to shift the mean of the estimated \gls{csi} away from the reflecting equivalent channel, $\mathbf{h}_{\mathrm{R}, k}$. Since the receive combining vector, $\mathbf{v}_k$, is a function of the estimated \gls{csi}, this shift will also inevitably influence it. Consequently, because of the linearity of the \gls{mr} combiner, we can express $\mathbf{v}_k$ as:
\begin{equation}
    \mathbf{v}_k(\varpi) = \bar{\mathbf{v}}_{{\rm P},k} + (1 - \varpi) {\mathbf{v}}_{{\rm R},k},
    \label{eq:distorted-received-combining}
\end{equation}
where $\mathbf{v}_k(\varpi)$ stresses that $\mathbf{v}_k$ is a function of the relative probe duration $\varpi$ (Def.~\ref{definition:varpi}) with $\bar{\mathbf{v}}_{{\rm P},k}$ representing the part of the receive combing vector that is misled by the probe distortion and ${\mathbf{v}}_{{\rm R},k}$ being the desired part from the point of view of correctly spatially separating the \glspl{ue}. Thus, the correspondent instantaneous \gls{ul} \gls{sinr} of~\eqref{eq:estimated-signal} can be written as 
\begin{equation}
    \mathrm{SINR}^{\rm UL}_{k}(\varpi)=\frac{
        \abs{\mathbf{v}_k^{\htransp}{\mathbf{h}}_{\mathrm{R}, k}s_k}^2
        }
        {
        \sum_{\substack{i\in\mathcal{K} \\ i\neq k}}\abs{\mathbf{v}_k^{\htransp}{\mathbf{h}}_{\mathrm{R},i}s_i}^2+\abs{\mathbf{v}_k^{\htransp}\mathbf{o}}^2
        }
\end{equation}
and the instantaneous \gls{ul} \gls{se} can be calculated as
%\begin{equation}
$
    \mathrm{SE}^{\rm UL}_k(\varpi)=\frac{\tau_u}{\tau_\chest+\tau_u}\log_2\left(1+\mathrm{SINR}^{\rm UL}_{k}\right),
$
%\end{equation}
whose quantities inherent the dependence on $\varpi$ from~\eqref{eq:distorted-received-combining}. We then apply the \gls{uatf} bound to more accurately estimate the \gls{hris}-assisted communication performance, as summarized below and adapted from~\cite[p.~302]{massivemimobook}.

\begin{corollary}[HRIS-assisted communication performance]
    The \gls{ul} \gls{se} of the $k-$th \gls{ue} can be lower bounded on average \gls{wrt} signal/noise realizations as 
    \begin{equation}
        \underline{\mathrm{SE}}^{\rm UL}_k(\varpi)=\frac{\tau_u}{\tau_\chest+\tau_u}\log_2\left(1+\underline{\mathrm{SINR}}^{\rm UL}_{k}(\varpi)\right),\,\text{where}
        \label{eq:mmimo:se-uatf}
    \end{equation}
    \vspace{-6pt}
    \begin{equation}
        \underline{\mathrm{SINR}}^{\rm UL}_{k}(\varpi)\!=\!\frac{\rho \mathbb{E}\left\{\abs{\mathbf{v}_k^{\htransp} {\mathbf{h}}_{\mathrm{R}, k}}^2 \right\} }{ \rho\! \sum_{i=1,i\neq k}^{K} \mathbb{E}\left\{\abs{\mathbf{v}_k^{\htransp} {\mathbf{h}}_{\mathrm{R}, i}}^2\right\} + \mathbb{E}\left\{\abs{\mathbf{v}_k^{\htransp} {\mathbf{o}}}^2\right\}}
    \end{equation}
    for $k\!\in\!\mathcal{K}$. Again, $(\varpi)$ stresses the dependence on these quantities on the relative probe duration inherited from~\eqref{eq:distorted-received-combining}, {acknowledging the effects of implementation complexity, which arise from the different hardware architectures, and of probe distortion, which arises from the probe-distorted \gls{csi} and is influenced by the specific hardware designs.}
    \label{corollary:uatf-ul-se}
\end{corollary}

The above corollary summarizes the HRIS-assisted communication performance under implementation complexity and probe distortion. However, the impact of probe distortion remains sternly hidden. To gain further insights, we extend the analysis to examine how~\eqref{eq:distorted-received-combining} influences the above result, aiming to demonstrate that probe distortion can ultimately reduce communication performance. While the impact manifests in the \gls{sinr}, our focus shifts to the \gls{sir} for convenience. By applying~\eqref{eq:distorted-received-combining} and leveraging the triangle inequality, the \gls{sir} of the $k-$th \gls{ue} is given by
\begin{equation}
    \underline{\mathrm{SIR}}^{\rm UL}_{k}\!\leq\!
    \frac{ 
        \mathbb{E}\!\left\{\abs{\bar{\mathbf{v}}_{\rm{P}, k}^{\htransp}{\mathbf{h}}_{\mathrm{R}, k}}^2\right\} 
        + (1-\varpi)^{2}
        \mathbb{E}\!\left\{\abs{\mathbf{v}_{\rm{R}, k}^{\htransp}{\mathbf{h}}_{\mathrm{R}, k}}^2\right\} 
    }
    {
        \sum_{\substack{i=1\\i\neq k}}^{K} \mathbb{E}\!\left\{\abs{\bar{\mathbf{v}}_{\rm{P}, k}^{\htransp} {\mathbf{h}}_{\mathrm{R}, i}}^2\right\}
        + (1-\varpi)^{2}
        \sum_{\substack{i=1\\i\neq k}}^{K} \mathbb{E}\!\left\{\abs{\mathbf{v}_{\rm{R}, k}^{\htransp} {\mathbf{h}}_{\mathrm{R}, i}}^2\right\}
    }.
    \label{eq:sir1}
\end{equation}
In the ideal case of zero probe distortion, this \gls{sir} is:
\begin{equation}
    \underline{\mathrm{SIR}}^{\rm UL}_{k}\!\leq\!
    { 
        \mathbb{E}\left\{\abs{\mathbf{v}_{\rm{R}, k}^{\htransp}{\mathbf{h}}_{\mathrm{R}, k}}^2\right\} 
    }
    /
    {
        \sum_{\substack{i=1,i\neq k}}^{K} \mathbb{E}\left\{\abs{\mathbf{v}_{\rm{R}, k}^{\htransp} {\mathbf{h}}_{\mathrm{R}, i}}^2\right\}
    }.
    \label{eq:sir2}
\end{equation}
To show that the probe distortion can be \textit{unfavorable}, that is, it \textit{can reduce} the \gls{hris}-assisted communication performance, we compare~\eqref{eq:sir1} to~\eqref{eq:sir2} and obtain the following result.

\begin{corollary}[Unfavorable probe distortion]
    The probe distortion can reduce the numerator of the \gls{sir} in~\eqref{eq:sir1} while simultaneously increasing its denominator. That is, probe distortion can elevate the interference power among \glspl{ue} while decreasing their effective power; reducing, rather than increasing, the overall \gls{sir}, which is the primary motivation of deploying an \gls{hris}. This effect is stochastic, as it depends on factors such as the positions of \glspl{ue}, channel realizations, and implementation complexity (\gls{hris} hardware architecture).
    \label{corollary:effect-of-probing}
\end{corollary}
\begin{proof}
    The proof can be found in Appendix~\ref{appx:proof-effect}.
\end{proof}

The above result along with Corollaries~\ref{corollary:power:performance} and~\ref{corollary:signal:performance} provide us ways to gain insights about the autonomous \gls{ris} trade-off and the underlying robust feasibility region. \textit{Note that we focused on showing that probing distortion can be \textbf{unfavorable} to communication performance, but this does not rule out the possibility of the opposite.}

%----------
\subsection{{Non-Autonomous \textit{vs.} Autonomous RISs}}\label{sec:network:non-vsautonomous}
%----------
\noindent
We briefly compare non-autonomous and autonomous \glspl{ris}, informed by the operational insights discussed above. Regarding the complexity of the \gls{chest} procedure at the \gls{bs}, following end-to-end \gls{chest} protocols, controlled \glspl{ris} require the \gls{bs} to estimate channel responses for all \gls{bs}-\gls{ris}
and \gls{ris}-\gls{ue} channels to optimize the \gls{ris} configuration~\cite{bjornson2021signalprocessing}. Without leveraging specific channel properties, such as channel sparsity, this involves estimating $K_{\max} (M\!+\!N\!+\!M N)$ channel responses (see~\cite{Wang2020:ce} for details). Alternatively, for \glspl{hris}, the oblivious \gls{bs} only needs to estimate $K_{\max} M$ channel responses, as indicated in Corollary~\ref{corollary:ls-channel-estimation}, resulting in significant computational savings for the \gls{bs}. In terms of overhead represented by $\tau_\chest$ in Corollary~\ref{corollary:uatf-ul-se}, an \gls{hris} requires the transmission of $L K_{\max}$ pilot samples due to pilot repetition, as detailed in Section~\ref{sec:framework:detailed}. In contrast, a controlled \gls{ris} requires dedicated and explicit control to receive its reflection configuration from the \gls{bs}, adding control overhead beyond that already needed for \gls{csi} acquisition~\cite{Zappone2021, Saggese2023,Saggese2023:ris-mec}. To model this as simply yet comprehensively as possible, we consider that
$
    \tau_\chest\! = \! 2K_{\max} + N + {N}/{R},
$ where $2K_{\max}\! +\! N$ represents the minimal number of pilot samples needed for end-to-end \gls{chest} estimation, as specified in~\cite{Wang2020:ce}, and $N/R$ models a basic yet unrealistic control channel capable of transmitting phase shifts with infinite precision and no errors at a rate $R\!>\!0$, measured in phase shifts per sample, similar to the model in~\cite{Zappone2021}. In the case of $R\!\xrightarrow[]{}\!\infty$, we have \textit{ideal control} with zero errors and overhead. In simple terms, the potential advantage of autonomous \gls{ris} over its non-autonomous counterpart lies in achieving $L K_{\max}\!>\!2K_{\max}\! + \!N \!+\! N/R$ along with comparing the respective assisted communication performances. For large $N$, the first condition is readily met. In the next section, we demonstrate that gains can still be obtained even under less favorable conditions for autonomous \glspl{ris}, specifically for $N$ in the few dozens.

%============================
\section{Experiments and Discussion}\label{sec:results}
%============================
\noindent
We numerically evaluate and discuss the fundamental trade-offs posed by autonomy, as defined in Section~\ref{sec:intro:contributions}.\footnote{Simulation code is available online at \href{https://github.com/victorcroisfelt/self-configuring-orchestration}{this link}.\label{foot:code}} Table~\ref{tab:parameters} reports the simulation parameters used, motivated by a suburban setting that uses the \gls{hris} to \textit{extend coverage} to \glspl{ue} in cell-edge conditions~\cite{massivemimobook,DiRenzo2020}. The \gls{hris} is located at the origin of a two-dimensional Cartesian system. The \gls{bs} is placed at the second quadrant $1$ km away from the \gls{hris} at $135^{\circ}$. The \glspl{ue} are randomly placed within a ring at the first quadrant with a respective inner and outer radius of $900$~m and $1$~km, representing a cell-edge condition. The \gls{bs} receiver noise $\sigma^2_{\rm B}\!=\!-94$~dBm comprises the thermal noise over $B\!=\!20$~MHz and a noise figure of $7$~dB in the receiver hardware; whereas, the \gls{hris} hardware has worse quality with a noise figure of~$10$ dB, yielding~$\sigma_{\rm H}^{2}\!=\!-91$~dBm. The above choices ensure strong enough \gls{bs}-\gls{ue} channels, assuring their spatial separability.

\begin{table*}[!t]
    \caption{
        Simulation parameters.
    }
    \label{tab:parameters}
    \centering
    \resizebox{\textwidth}{!}{%
    \scriptsize
    \begin{tabular}{l!{\color{white}\vrule}c|l!{\color{white}\vrule}c|l!{\color{white}\vrule}c}
        \toprule
        \textbf{Parameter} & \textbf{Value} & \textbf{Parameter} & \textbf{Value}  & \textbf{Parameter} & \textbf{Value}\\  
        \midrule
        \rowcolor[HTML]{EFEFEF}
        Carrier frequency, $f_c$ & $28$ GHz & Coupling parameter, $\eta$ & $0.999$ & Max. number of \glspl{ue}, $K_{\max}$ & $4$ \\
        \gls{bs} pathloss, $\beta_{\rm B}$ & $3.76$ & Receiver noise powers, $\sigma^2_{\rm H},\sigma^2_{\rm B}$ & $-91$,$-94$ dBm & Coherence block length, $\tau_c$ & $128$ samples \\
        \rowcolor[HTML]{EFEFEF}
        \gls{hris} pathloss, $\beta_{\rm H}$ & $2$ & \gls{ul} transmit power, $\rho$ & $0$ dBm & Number of pilots, $\tau_p$ & $K_{\rm max}$ \\
        Channel power gain, $\gamma_0$ & $1$ & Number of \gls{bs} antennas, $M$ & $64$ & Number of pilot subblocks, $L$ & 16\\
        \rowcolor[HTML]{EFEFEF}
        \gls{nlos} relative powers, $\sigma^2_{\rm DR}$,\,$\sigma^2_{\rm RR}$ & $90.8\times10^{-9},\,0.11\times10^{-6}$ & Number of \gls{hris} elements, $N$ & $32$ & Phase durations, $\tau_\chest,\tau_\comm=\tau_u$ & $64,64$ \\
        \bottomrule
    \end{tabular}
    }
\end{table*}
%

%----------
\subsection{Implementation Complexity Trade-Off}\label{sec:results:hris}
%----------

\paragraph*{\textbf{Probing performance}} For this evaluation, we assume that each \gls{ue} has a 50\% probability of being scheduled within a given coherence block. Figure~\ref{fig:results:hris-probe-performance} shows the probing performance in terms of the probability of detection, $P_{\rm D}$, for different choices of probabilities of false alarm, $P_{\rm FA}$. We use Corollaries~\ref{corollary:power:performance} and~\ref{corollary:signal:performance} to determine threshold values, and Monte Carlo simulations with $10^4$ realizations to obtain the curves. Naturally, we observe degradation in $P_{\rm D}$ with a decrease in the power absorbed by the \gls{hris}, $1\!-\!\eta$, or a decrease in the relative probe duration, $\varpi$ (see Def.~\ref{definition:varpi}). Furthermore, as expected, the \gls{dsp}-enabled \gls{hris} outperforms the \gls{pd}-enabled counterpart; however, this performance gap narrows as $1\!-\!\eta$ increases. Notably, we also note fluctuations in the performance of the \gls{pd}-enabled \gls{hris}, which arise from our design choice to alter the probing configuration codebook as a function of $C$ or $\varpi\!=\!C/L$, as indicated in~\eqref{eq:hris:probing:power:sweeping}. Based on Fig.~\ref{fig:results:hris-probe-performance}, we select $\eta\!=\!0.999$ and $P_{\rm FA}\! =\! 10^{-2}$ for the next simulations, enabling us to focus on evaluating the reflection mode with a satisfactory probing performance{, which is around $P_{\rm D}\!=\!93.14$\% for the \gls{pd}- and $P_{\rm D}\!=\!99.57$\% for the \gls{dsp}-enabled \gls{hris}.} Hence, when manufacturing the \gls{hris}, the choice of the coupling parameter $\eta$ can be aligned with the desired performance for both the probe and reflection modes.

\paragraph*{\textbf{Reflecting performance}}
To isolate the reflecting performance to not depend on scheduling, we consider $K\!=\!2$ \glspl{ue} that are \textit{always} scheduled. Figure~\ref{fig:results:hris-reflection-performance} shows the reflecting performance in terms of the \gls{hris}-\gls{ue} channel gain. Here, the selection of $C\!=\!8$ indicates that the probe mode occupies half of the \gls{chest} phase, $\varpi\!=\!0.5$. As anticipated, the \gls{dsp}-enabled \gls{hris} demonstrates superior performance compared to its \gls{pd}-enabled counterpart, as it is more effective in localizing the \glspl{ue} and reflecting energy toward them, resulting in higher average channel gains. In terms of the $\mathrm{NMSE}_{\rm H}$, defined in~\eqref{eq:hris:reflection:nmse}, the \gls{pd}- achieves $0.86$ while $1.72\times10^{-6}$ is achieved by the \gls{dsp}-enabled \gls{hris}, showing that \gls{dsp} capabilities plays a huge difference in getting more accurate local \gls{csi}.

\begin{figure}[!t]
    \centering
    % This file was created with tikzplotlib v0.10.1.
\begin{tikzpicture}

\definecolor{darkorange25512714}{RGB}{255,127,14}
\definecolor{forestgreen4416044}{RGB}{44,160,44}
\definecolor{lavender233}{RGB}{233,233,233}
\definecolor{steelblue31119180}{RGB}{31,119,180}

\def\shift{-0.8cm}
\def\sep{1cm}
\def\vside{2.2cm}

\begin{groupplot}[
    % General format
    group style={
        group name=performance,
        group size=2 by 1,
        horizontal sep=0.4in,
        ylabels at=edge left
    }, 
    height=0.7in,
    width=1.25in,
    scale only axis,
    tick align=outside,
    tick pos=left,
    % ylabel shift={
    %     -6pt
    % },
    % tick label style={
    %     font=\scriptsize
    % },
    legend style={
        at={(1.15,1)},
        anchor=south,
        draw=none,
        fill=none,
        legend cell align=center,
        align=center
    },
    legend columns=-1,
    ylabel={
        Proba. of detection , \(\displaystyle P_{\rm D}\)
    },
    xmajorgrids,
    ymajorgrids,
]

\nextgroupplot[
    log basis x={10},
    tick align=outside,
    tick pos=left,
    xlabel={Fraction of pow. absorbed, \(\displaystyle 1\!-\!\eta\)},
    xmajorgrids,
    xmin=1e-04, xmax=0.1,
    xmode=log,
    xtick={1e-06,1e-05,0.0001,0.001,0.01,0.1,1,10},
    xticklabels={
        \(\displaystyle {10^{-6}}\),
        \(\displaystyle {10^{-5}}\),
        \(\displaystyle {10^{-4}}\),
        \(\displaystyle {10^{-3}}\),
        \(\displaystyle {10^{-2}}\),
        \(\displaystyle {10^{-1}}\),
        \(\displaystyle {10^{0}}\),
        \(\displaystyle {10^{1}}\)
    },
    ymin=0.575, ymax=1.025,
    ytick={0.6,0.7,0.8,0.9,1},
    yticklabels={
    \(\displaystyle {0.6}\),
    \(\displaystyle {0.7}\),
    \(\displaystyle {0.8}\),
    \(\displaystyle {0.9}\),
    \(\displaystyle {1.0}\)
    }
]
\addplot [darkorange25512714, mark=square*, mark size=1.5, mark options={solid}, forget plot]
table {%
    0.1000 0.9996
    0.0100 0.9984
    0.0010 0.9314
    0.0001 0.6843
};
\addplot [forestgreen4416044, dashed, mark=square*, mark size=1.5, mark options={solid}, forget plot]
table {%
    0.1000 1.0000
    0.0100 0.9968
    0.0010 0.8900
    0.0001 0.6005
};
\addplot [darkorange25512714, mark=diamond*, mark size=1.5, mark options={solid}, forget plot]
table {%
    0.1000 0.9957
    0.0100 0.9956
    0.0010 0.9957
    0.0001 0.9957
};
\addplot [forestgreen4416044, dashed, mark=diamond*, mark size=1.5, mark options={solid}, forget plot]
table {%
    0.1000 0.9996
    0.0100 0.9995
    0.0010 0.9996
    0.0001 0.9996
};

\addlegendimage{ycomb, gray, mark=square*, mark size=1.5, mark options={solid,fill=gray}}
\addlegendentry{\gls{pd}-enabled}

\addlegendimage{ycomb, gray, mark=diamond*, mark size=1.5, mark options={solid,fill=gray}}
\addlegendentry{\gls{dsp}-enabled}

\addlegendimage{darkorange25512714}
\addlegendentry{$P_\mathrm{FA} = 10^{-2}$}

\addlegendimage{forestgreen4416044, dashed}
\addlegendentry{$P_\mathrm{FA} = 10^{-3}$}

\nextgroupplot[
    xlabel={Rel. probe duration, \(\displaystyle \varpi\)},
    xmin=0, xmax=1,
    xtick={
        0,
        %0.1,
        0.2,
        %0.3,
        0.4,
        %0.5,
        0.6,
        %0.7,
        0.8,
        %0.9,
        1
    },
    xticklabels={
        \(\displaystyle {0.0}\),
        %\(\displaystyle {0.1}\),
        \(\displaystyle {0.2}\),
        %\(\displaystyle {0.3}\),
        \(\displaystyle {0.4}\),
        %\(\displaystyle {0.5}\),
        \(\displaystyle {0.6}\),
        %\(\displaystyle {0.7}\),
        \(\displaystyle {0.8}\),
        %\(\displaystyle {0.9}\),
        \(\displaystyle {1.0}\)
    },
    ymin=0.675, ymax=1.025,
    ytick style={color=black},
    ytick={0.7,0.8,0.9,1},
    yticklabels={
        %\(\displaystyle {0.4}\),
        %\(\displaystyle {0.5}\),
        %\(\displaystyle {0.6}\),
        \(\displaystyle {0.7}\),
        \(\displaystyle {0.8}\),
        \(\displaystyle {0.9}\),
        \(\displaystyle {1.0}\)
    }
]
\addplot [darkorange25512714, mark=square*, mark size=1.5, mark options={solid}]
table {%
    0.0625 0.7791
    0.1250 0.8374
    0.1875 0.8627
    0.2500 0.8791
    0.3125 0.9145
    0.3750 0.9072
    0.4375 0.9522
    0.5000 0.9314
    0.5625 0.9729
    0.6250 0.9688
    0.6875 0.9755
    0.7500 0.9852
    0.8125 0.9946
    0.8750 0.9931
    0.9375 0.9920
    1.0000 0.9958
};
\addplot [forestgreen4416044, dashed, mark=square*, mark size=1.5, mark options={solid}]
table {%
    0.0625 0.7061
    0.1250 0.7812
    0.1875 0.8133
    0.2500 0.8299
    0.3125 0.8758
    0.3750 0.8588
    0.4375 0.9229
    0.5000 0.8900
    0.5625 0.9480
    0.6250 0.9456
    0.6875 0.9571
    0.7500 0.9679
    0.8125 0.9909
    0.8750 0.9864
    0.9375 0.9833
    1.0000 0.9893
};
\addplot [darkorange25512714, mark=diamond*, mark size=1.5, mark options={solid}]
table {%
    0.0625 0.9957
    0.1250 0.9957
    0.1875 0.9957
    0.2500 0.9957
    0.3125 0.9958
    0.3750 0.9957
    0.4375 0.9957
    0.5000 0.9957
    0.5625 0.9957
    0.6250 0.9956
    0.6875 0.9958
    0.7500 0.9957
    0.8125 0.9956
    0.8750 0.9956
    0.9375 0.9956
    1.0000 0.9956
};
\addplot [forestgreen4416044, dashed, mark=diamond*, mark size=1.5, mark options={solid}]
table {%
    0.0625 0.9996
    0.1250 0.9996
    0.1875 0.9996
    0.2500 0.9996
    0.3125 0.9996
    0.3750 0.9996
    0.4375 0.9996
    0.5000 0.9996
    0.5625 0.9996
    0.6250 0.9996
    0.6875 0.9996
    0.7500 0.9995
    0.8125 0.9996
    0.8750 0.9996
    0.9375 0.9995
    1.0000 0.9995
};
\end{groupplot}

\node[text width=.49\columnwidth, align=center, anchor=north] at ([yshift=\shift]performance c1r1.south) {\subcaption{
    %Portion of power absorbed by the \gls{hris} w/ 
    $\varpi=0.5$ or $C=8$.
}};

%\node[text width=1\textwidth, align=center, anchor=north] at ([yshift=-\belowcaptionskip]xlabel_a) {\subcaption{$\varpi=0.5$ or $C=8$.}};

\node[text width=.49\columnwidth, align=center, anchor=north] at ([yshift=\shift]performance c2r1.south) {\subcaption{
    %Relative probe duration within the \gls{chest} phase w/ 
    $\eta=0.999$.
}};

\end{tikzpicture}
    \vspace{-18pt}
    \caption{
        Numerical comparison of the probing performance regarding implementation complexity for the \gls{pd}- and \gls{dsp}-enabled \gls{hris} hardware architectures. We vary (a) the fraction of power absorbed by the \gls{hris}, $1\!-\!\eta$, and (b) {the level of probe distortion} via the relative probe duration, $\varpi$ (Def.~\ref{definition:varpi}). We evaluate different choices of the probability of false alarm $P_{\rm FA}$ for $K\!=\!4$ \glspl{ue} with each having a probability of 50\% to be scheduled on each realization.
    }
    \label{fig:results:hris-probe-performance}
    \vspace{-16pt}
\end{figure}
\begin{figure}[!t]
    \centering
    % This file was created with tikzplotlib v0.10.1.
\begin{tikzpicture}

\definecolor{darkorange25512714}{RGB}{255,127,14}
\definecolor{forestgreen4416044}{RGB}{44,160,44}
\definecolor{steelblue31119180}{RGB}{31,119,180}

\def\shift{-0.7cm}
\def\sep{0.3cm}
\def\vside{2.2cm}

    %     horizontal sep=0.4in,
    %     ylabels at=edge left
    % }, 
    % height=0.8in,
    % width=1.25in,

\begin{groupplot}[
    % General format
    group style={
        group name=hris,
        group size=2 by 1,
        horizontal sep=0.1in,
        xlabels at=edge bottom,
        y descriptions at=edge left
    }, 
    height=0.7in,
    width=1.25in,
    scale only axis,
    title style={at={(0.5,0.82)}},
    tick style={color=black},
    tick align=outside,
    tick pos=left,
    xlabel={
        UE angular position [$^{\circ}$]
    },
    xmajorgrids,
    xmin=0, xmax=90,
    xtick style={color=black},
    xtick={0,30,60,90,120,150,180},
    xticklabels={
        \(\displaystyle {0}\),
        \(\displaystyle {30}\),
        \(\displaystyle {60}\),
        \(\displaystyle {90}\),
    },
    ylabel={
        HRIS-UE ch. gain
    },
    ymajorgrids,
    ymin=0, ymax=500,
    ytick={0,100,200,300,400,500},
    yticklabels={
        %\(\displaystyle {\ensuremath{-}50}\),
        \(\displaystyle {0}\),
        %\(\displaystyle {}\),
        \(\displaystyle {100}\),
        %\(\displaystyle {150}\),
        \(\displaystyle {200}\),
        %\(\displaystyle {250}\)
        \(\displaystyle {300}\),
        \(\displaystyle {400}\),
        \(\displaystyle {500}\)
    }
]
]

\nextgroupplot[]
\addplot [thick, black, forget plot]
table {%
0.0000 15.1716
0.9000 15.2705
1.8000 15.5634
2.7000 16.0380
3.6000 16.6710
4.5000 17.4234
5.4000 18.2355
6.3000 19.0202
7.2000 19.6582
8.1000 19.9952
9.0000 19.8472
9.9000 19.0152
10.8000 17.3191
11.7000 14.6541
12.6000 11.0746
13.5000 6.9072
14.4000 2.8814
15.3000 0.2586
16.2000 0.9181
17.1000 7.3446
18.0000 22.4493
18.9000 49.1624
19.8000 89.7683
20.7000 145.0262
21.6000 213.2218
22.5000 289.4157
23.4000 365.2450
24.3000 429.6312
25.2000 470.5939
26.1000 478.0381
27.0000 446.9235
27.9000 379.7985
28.8000 287.5225
29.7000 187.3319
30.6000 98.2943
31.5000 35.4148
32.4000 4.6526
33.3000 1.2180
34.2000 12.3535
35.1000 23.6378
36.0000 25.7648
36.9000 18.1243
37.8000 7.1327
38.7000 0.4999
39.6000 1.4060
40.5000 6.7153
41.4000 10.3395
42.3000 8.7369
43.2000 3.7292
44.1000 0.2080
45.0000 1.0764
45.9000 4.4863
46.8000 6.1455
47.7000 4.1372
48.6000 0.9036
49.5000 0.1571
50.4000 2.3792
51.3000 4.3699
52.2000 3.4327
53.1000 0.8603
54.0000 0.1071
54.9000 2.0176
55.8000 3.6689
56.7000 2.5867
57.6000 0.3877
58.5000 0.3925
59.4000 2.5472
60.3000 3.4094
61.2000 1.5501
62.1000 0.0020
63.0000 1.4549
63.9000 3.6276
64.8000 2.8084
65.7000 0.3319
66.6000 0.8032
67.5000 4.0260
68.4000 4.5315
69.3000 1.1594
70.2000 0.5522
71.1000 6.0088
72.0000 9.3426
72.9000 3.6486
73.8000 0.9155
74.7000 22.8421
75.6000 69.4223
76.5000 108.5507
77.4000 107.3859
78.3000 68.4548
79.2000 24.7493
80.1000 2.6526
81.0000 0.5704
81.9000 2.8994
82.8000 1.9127
83.7000 0.1761
84.6000 0.2231
85.5000 0.6874
86.4000 0.3575
87.3000 0.0041
88.2000 0.1468
89.1000 0.2391
90.0000 0.0658
};

\addplot [ycomb, red, dotted, very thick]
table {%
30  0
30  500
75  0
75  500
};

% %% ADD USERS
% \addplot [ycomb, red, dotted, mark=x, mark size=3.5, mark options={solid}]
% table {%
% 30  75
% 80  150
% 110 150
% 150 75
% };

% % Nodes and annotation
% \node[align=center, anchor=center, font=\scriptsize, name=poor] at (56.5,150) {poorly\\served};
% % \draw[rounded corners, dashed] (70, 140) rectangle (90,160);
% \draw [->] ([yshift=0cm]poor.east) -- (75,150);

% \node[align=center, anchor=center, font=\scriptsize, name=well] at (56.5,75) {well\\served};
% % \draw [->] ([xshift=-0.1cm]well.east) -- (140,75);
% \draw [->] ([xshift=0cm]well.west) -- (35, 75);

\nextgroupplot[ymajorticks=false]
\addplot [thick, black]
table {%
0.0000 1.5224
0.9000 1.4767
1.8000 1.3429
2.7000 1.1307
3.6000 0.8588
4.5000 0.5567
5.4000 0.2690
6.3000 0.0582
7.2000 0.0079
8.1000 0.2228
9.0000 0.8250
9.9000 1.9429
10.8000 3.6906
11.7000 6.1368
12.6000 9.2624
13.5000 12.9122
14.4000 16.7506
15.3000 20.2411
16.2000 22.6732
17.1000 23.2681
18.0000 21.3856
18.9000 16.8382
19.8000 10.2788
20.7000 3.5744
21.6000 0.0186
22.5000 4.1963
23.4000 21.3255
24.3000 56.0005
25.2000 110.4636
26.1000 182.8064
27.0000 265.7562
27.9000 346.8022
28.8000 410.2025
29.7000 440.8314
30.6000 428.9819
31.5000 374.4605
32.4000 288.0599
33.3000 189.1400
34.2000 99.5858
35.1000 36.2577
36.0000 5.2314
36.9000 0.6977
37.8000 9.2386
38.7000 17.2651
39.6000 17.4928
40.5000 10.8711
41.4000 3.3134
42.3000 0.0244
43.2000 1.6658
44.1000 4.8486
45.0000 5.8284
45.9000 3.8065
46.8000 1.0174
47.7000 0.0108
48.6000 1.1090
49.5000 2.4932
50.4000 2.4484
51.3000 1.1349
52.2000 0.0792
53.1000 0.2626
54.0000 1.1216
54.9000 1.4444
55.8000 0.8628
56.7000 0.1341
57.6000 0.0756
58.5000 0.6014
59.4000 0.9616
60.3000 0.7108
61.2000 0.1819
62.1000 0.0160
63.0000 0.4093
63.9000 0.9087
64.8000 0.9131
65.7000 0.3702
66.6000 0.0013
67.5000 0.6719
68.4000 2.2033
69.3000 2.8955
70.2000 1.3375
71.1000 0.1102
72.0000 7.2214
72.9000 30.3780
73.8000 66.1671
74.7000 95.8958
75.6000 97.9809
76.5000 68.7389
77.4000 28.6801
78.3000 3.7404
79.2000 0.9312
80.1000 6.5003
81.0000 6.5195
81.9000 1.7634
82.8000 0.1112
83.7000 2.4509
84.6000 3.3100
85.5000 1.1819
86.4000 0.0209
87.3000 1.3896
88.2000 2.1989
89.1000 0.9048
90.0000 0.0044
};

%% ADD USERS
\addplot [ycomb, red, dotted, very thick]
table {%
30  0
30  500
75  0
75  500
};

% %% ADD USERS
% \addplot [ycomb, red, dotted, mark=x, mark size=3.5, mark options={s}olid]
% table {%
% 30  40
% 80  75
% 110 75
% 150 40
% };

% Nodes and annotation
% \node[align=center, anchor=center, font=\scriptsize, name=well] at (45,150) {well\\served};
% \draw [->] ([xshift=-0cm]well.south) -- (80,90);
% \draw [->] ([xshift=0cm]well.south) -- (30, 50);
% \draw [->] ([yshift=-0cm]well.north) -- (95,68);
% \draw[rounded corners, dashed] (70, 70) rectangle (120,80);

\end{groupplot}

%% Subcaptions
% \node[text width=.32\textwidth, align=center, anchor=north] at ([yshift=\shift]hris c1r1.south) {\subcaption{\label{fig:results:hris-pd} Worst-case probability of detection}};
% \node[text width=.32\textwidth, align=center, anchor=north] at ([yshift=\shift]hris c1r1.south) {\subcaption{\label{fig:results:hris:power-reflection} PD-enabled}};
% \node[text width=.32\textwidth, align=center, anchor=north] at ([yshift=\shift]hris c2r1.south) {\subcaption{\label{fig:results:hris:signal-reflection} DSP-enabled}};

\end{tikzpicture}
    \vspace{-6pt}
    \caption{
        Qualitative comparison of the reflecting performance regarding implementation complexity for the (\textit{left}) \gls{pd}- and (\textit{right}) \gls{dsp}-enabled \gls{hris} hardware architectures. We assume $K\!=\!2$ always-scheduled \glspl{ue} for $\eta\!=\!0.999$, $P_{\rm FA}\!=\!10^{-2}$, $K_{\max}\!=\!4$, and $\varpi\!=\!0.5$ or $C\!=\!8$. To enhance visualization, the \gls{bs} is placed $1$ km from the normal line to the \gls{hris}. The `\textcolor{red}{$\cdots$}' lines represent the $2$ \glspl{ue} positioned at $(d_k,\theta_k)$:$(10 \text{ m}, 30^{\circ})$ and $(20 \text{ m},75^{\circ})$.
    }
    \label{fig:results:hris-reflection-performance}
    \vspace{-12pt}
\end{figure}

\paragraph*{\textbf{Overall \gls{hris} performance}} 
We now summarize the key findings regarding the trade-off under evaluation. Based on the complexity analysis in Section~\ref{sec:hris:complexity}, the \gls{pd}-enabled \gls{hris} requires $308$ element-wise operations, while the \gls{dsp}-enabled variant requires $3264$ for $\eta\!=\!0.999$, $P_{\rm FA}\!=\!10^{-2}$, and $\varpi\!=\!0.5$. Hence, the \gls{pd}-enabled \gls{hris} could achieve a computational saving of approximately $90.56$\%, where likely similar gains are expected in capital costs and energy consumption (this very much depends on how the analog circuitry of the \gls{pd}-enabled \gls{hris} is implemented). Although the detection performance difference between the two architectures is only $6.43$\%, the \gls{dsp}-enabled \gls{hris} offers significantly higher quality in local \gls{csi} at the \gls{hris}. In this part, however, we have focused on evaluating the operation modes in isolation. Next, we will assess a more practical scenario of interest that considers the impact of the \gls{hris} operation on communication performance. 

%----------
\subsection{Autonomous RIS Trade-Off}\label{sec:results:ablation}
%----------

\paragraph*{\textbf{Baselines}} We consider three baselines. \textbf{Standalone} refers to an \gls{mmimo} system that operates independently, without the assistance of the \gls{hris}. \textbf{Informed \gls{bs}} represents an informed \gls{bs} employing the stop-and-wait approach to avoid probe distortion while the \gls{hris} is ideal with perfect probing and reflection performance. Also, we do not account for additional overhead needed for the \gls{bs} to be aware of the \gls{hris} operation. \textbf{Controlled RIS} denotes the non-autonomous \gls{ris} paradigm characterized by the control overhead outlined in Section~\ref{sec:network:non-vsautonomous}, where $R\!\xrightarrow[]{}\!\infty$ signifies ideal control while $R\!=\!1$ indicates that one phase shift can be transmitted per sample with infinite precision and no errors. Also, we assume that the \gls{bs} has perfect \gls{csi}. \textit{Note:} It is important to recognize that the Informed \gls{bs} and Controlled \gls{ris} serve as highly optimistic baselines and their comparison with \gls{hris}-assisted performance is invariably unfair; but even so the latter show comparable performance while completely avoiding the need of dedicated, explicit control, as seen next.

\paragraph*{\textbf{Impact of the probe distortion under different levels of implementation complexity}} Figure~\ref{fig:results:network-performance} shows the performance of an \gls{hris}-assisted \gls{mmimo} system in terms of the quality of the probe-distorted \gls{csi} estimated at the \gls{bs}, given in Corollary~\ref{corollary:ls-channel-estimation}, and the \gls{se}, given in Corollary~\ref{corollary:uatf-ul-se}. To isolate the effect of probe distortion, we assume that the $K\!=\!4$ \glspl{ue} are always scheduled. The level of probe distortion can be controlled by increasing the probe relative duration, $\varpi$ (Def.~\ref{definition:varpi}). The higher $\varpi$, the better the \gls{hris} probe performance, but the worse the quality of the \gls{csi} obtained at the \gls{bs}. As expected, this is readily seen in the \gls{nmse} curves shown in Fig.~\ref{fig:results:network-performance} (left). From the \gls{se} curves, we observe that the Informed \gls{bs} achieves 0.3254 bits/s/Hz/UE, while the highest \glspl{se} for the \gls{pd}- and the \gls{dsp}-enabled \gls{hris} are 0.3251 and 0.3254 bits/s/Hz/UE, respectively, achieved at $\varpi\!=\!0.625$ ($C\!=\!10$) and $\varpi=0.0625$ ($C\!=\!1$). \textit{The negligible performance gap demonstrates that keeping the \gls{bs} oblivious of \gls{hris} operations does not significantly impact network performance when the proposed orchestration framework is applied.}

\begin{figure*}[!t]
    \centering
    \input{figs/figure9}
    \vspace{-6pt}
    \caption{
        Performance of an \gls{hris}-assisted \gls{mmimo} system with $K=4$ always-scheduled \glspl{ue} for $\eta\!=\!0.999$ and $P_{\rm FA}\!=\!10^{-2}$. Adjusting the relative probe duration, $\varpi$, allows us to control probe distortion, with $\varpi=1$ indicating maximum distortion. The two different hardware architectures characterize the two extremes of implementation complexity: \gls{pd}- is the lowest while \gls{dsp}- is the highest. For $\varpi\!=\!0$, the \gls{hris}-related curves have performance equal to that of the standalone \gls{mmimo} system, as this virtually means that the \gls{hris} is turned off.
    }
    \label{fig:results:network-performance}
\end{figure*}

\paragraph*{\textbf{Dual effect of probe distortion}} 
As outlined in Section~\ref{sec:intro}, our goal is to highlight the effects of probe distortion and explore its potential impact on \gls{hris}-assisted communication performance. Figure~\ref{fig:results:network-performance} shows that, while probe distortion harms the quality of \gls{csi} acquisition at the \gls{bs}, it only slightly influences \gls{se} performance. Specifically, even though the \gls{se} achieved by the \gls{dsp}-enabled \gls{hris} decreases when $\varpi>0.5$, it never falls below the performance of the Standalone system. \textit{This leads to the counterintuitive observation: in some cases, probe distortion may be \textbf{favorable}, that is, it can preserve or improve communication performance, even with the \gls{csi} quality at the \gls{bs} deteriorating.} One possible explanation for this phenomenon is that probe distortion may introduce diversity among the \gls{ue} channels without compromising their identity, thereby reducing interference. This is similar to the effect induced by spatial correlation~\cite{massivemimobook} and channel rank enhancement~\cite{DiRenzo2020}. Another explanation that supports the former is that we are analyzing a cell-edge condition, where \gls{csi} quality does not matter much as the received power is very low. Furthermore, probe distortion has a more significant impact on the \gls{dsp}-enabled \gls{hris} \gls{se} performance than on the \gls{pd}-enabled \gls{hris}, where \gls{se} is not affected by increases in $\varpi$. A possible explanation is that the \gls{pd}-enabled \gls{hris} produces a broader reflecting beam, distributing energy more evenly across the space; conversely, the \gls{dsp}-enabled \gls{hris} further narrows the energy focusing, leading to more prominent errors in \gls{csi} acquisition. \textit{Interestingly, our results suggest that architectures and algorithms that depend on lower \gls{dsp} capabilities can be advantageous in the presence of favorable probe distortion and in scenarios where \gls{csi} quality has a lower impact.} However, a more detailed understanding of this dual phenomenon of favorable and unfavorable probing distortion is required, as it is highly dependent on multiple factors, such as the \gls{hris} hardware architecture, the receive combining scheme, and the deployment setting. From~Fig.~\ref{fig:results:network-performance}, the \textit{robust feasibility region} can be visually characterized by finding values of $\varpi$ in which the \gls{hris}-related \gls{se} curves perform better than the Standalone baseline. For the cell-edge setting, the region is $\varpi\!\in\![0.0625, 1]$, with gains up to $19.06$\% and $20.29$\% on average for the \gls{pd}- and \gls{dsp}-enabled \glspl{hris}, respectively. We report that this region eventually narrows as \glspl{ue} are brought closer to the \gls{bs}.

\paragraph*{\textbf{Autonomous-vs-non-autonomous \glspl{ris}}} By comparing the performance of the controlled RIS with that of the \gls{hris}, we observe that the \gls{hris} generally performs worse than the controlled RIS under ideal control conditions ($R\!\rightarrow\!\infty$). But, when accounting for control overhead, the \gls{hris} can offer comparable or even superior performance depending on the value of $R$. Notably, this advantage comes with the benefit of not requiring dedicated, explicit control, which can be more costly than manufacturing and designing the \gls{hris} itself. 

%============================
\section{Conclusions}\label{sec:conclusions}
%============================
\noindent
We proposed a \gls{phy}-layer orchestration framework that aligns \gls{hris} operation modes with \gls{mmimo} operation phases, enabling the study of fundamental trade-offs concerning two major challenges posed by \glspl{ris} featuring autonomy: implementation complexity and probe distortion. As stringent conditions present in our analysis, we consider (a) two extremes of implementation complexity, realized by minimal \gls{hris} operation designs over the \gls{pd}- and \gls{dsp}-enabled \gls{hris} hardware architectures, and (b) an oblivious \gls{bs} that fully embraces probe distortion. Regarding the implementation complexity trade-off, our results showed that the more complex \gls{dsp}-enabled \gls{hris} has clearly better local \gls{csi} quality, but the \gls{pd}-enabled \gls{hris} can counterintuitively outperform it in terms of communication performance due to \textit{more favorable} probe distortion when supporting cell-edge \glspl{ue}. Regarding the autonomous \gls{ris} trade-off, we observed that unfavorable probe distortion can degrade \gls{hris}-assisted communication performance, potentially making autonomous \glspl{ris} unfeasible if not properly designed. However, we also observed a \textit{dual effect of probe distortion}, which can be favorable or unfavorable depending on several factors. Further research into the statistical properties of probe distortion is necessary to better understand this dual phenomenon. For example, we have conducted preliminary simulations that show that probe distortion can behave differently depending on the receive combining scheme being used; you can use our simulation platform to test it yourself for the \gls{zf} scheme.

In summary, we presented empirical evidence that an \gls{hris}-assisted \gls{mmimo} system can outperform standalone \gls{mmimo} and controlled \gls{ris} systems even under stringent conditions. Future research can expand this analytical framework to scenarios where the \gls{hris} supports multiple operators or \glspl{bs}. Additionally, it could incorporate performance analysis of hybrid controlled/autonomous \glspl{ris}, where some explicit control messages guide \gls{hris} behavior in some coherence blocks while allowing autonomous operation in others.

%============================
% Appendices
%============================
\appendices

%============================
\section{Proof of Corollary \ref{corollary:power:performance}}\label{appx:proof-power}
%============================
\begin{proof}
    We need to get the distributions of the numerator and the denominator of the left-hand side term of~\eqref{eq:hris:probing:power:general-glrt}. We start with the denominator. For the null-hypothesis in~\eqref{eq:hris:probing:power:hypothesis} with $A_k[c]\!=\!0$, $\alpha_k[c]\!=\!|n_t[c]|^{2}$ is distributed as an exponential distribution. Specifically,
    $p(\alpha_{t}[c];\mathcal{H}^{(k)}_0[c])\!=\!\mathrm{Exp}({1}/(2{N\sigma^2_\mathrm{H}}))$,
    where $\sigma^2_{\rm H}$ is a known nuisance parameter. For the numerator, the signal under $\mathcal{H}^{(k)}_1$ is approximated as $\alpha_{t}[c] \!\approx\!\abs{A_{k}[c]}^2 + \abs{{{n}}_{t}[c]}^2$, motivated by analyzing the signal on expectation, resulting in $2\Re\{A_{k}[c] {n}_{t}[c]\}$ being $0$ since the noise has zero mean. Note that its variance is still preserved in the term $\abs{{{n}}_{t}[c]}^2$. This approximation will surely cause an overestimation of the performance since we ignore the cross-term mixing amplitude and noise. Another motivation for such an approximation is to note that the terms $\abs{A_{k}[c]}^2$ and $\abs{{{n}}_{t}[c]}^2$ would be higher in magnitude than $2\Re\{A_{k}[c] {n}_{t}[c]\}$, where for high \gls{snr} values $\abs{A_{k}[c]}^2$ dominates; in contrast, $\abs{{{n}}_{t}[c]}^2$ dominates in low \gls{snr}. Hence, the numerator of~\eqref{eq:hris:probing:power:general-glrt} is distributed as
    $ 1/({2N\sigma^2_\mathrm{H}})\exp(-{1}/({2N\sigma^2_\mathrm{H}})\left(\alpha-{f}_\mathrm{LS}(A_{k}[c])\right))$.
    By the above and~\eqref{eq:hris:probing:power:general-glrt}, the \gls{hris} decides that the $k-$th \gls{ue} is detected in the $c-$th pilot subblock if
    $ \alpha_{t}[c]\!\gtrsim\!{2N\sigma^2_\mathrm{H}}\epsilon_s\!=\!\epsilon^\prime_s$.
\end{proof}

%============================
\section{Proof of Corollary~\ref{corollary:effect-of-probing}}\label{appx:proof-effect}
%============================
\begin{proof}
    Let the four terms that compose the $\underline{\mathrm{SIR}}^{\rm UL}_{k}$ in \eqref{eq:sir1} be referred to as:
    $ \underline{\mathrm{SIR}}^{\rm UL}_{k}\!=\! ( a +\varpi^{2} b )/ (  c +\varpi^{2} d  )
    $.
    To support our claim that probing distorting can be detrimental, we need to show that $a\!\leq \!b$ while $c\!\geq\! d$ for arbitrary choices of $\bar{\mathbf{v}}_{\rm{P},k}$, $\bar{\mathbf{v}}_{\rm{R},k}$, ${\mathbf{h}}_{\rm{R},k}$, ${\mathbf{h}}_{\rm{R},i}$. We must work with at least $M\!\geq\! 2$. For the sake of argument, we choose $\bar{\mathbf{v}}_{\rm{P},k}\!=\![0, 1]^{\transp}$, ${\mathbf{v}}_{\rm{R},k}\!=\![1, 0]^{\transp}$, ${\mathbf{h}}_{\rm{R},k}\!=\![0, 1]^{\transp}$, and ${\mathbf{h}}_{\rm{R},i}\!=\![1, 0]^{\transp}$. This yields in $a\!=\!0$, $b\!=\!1$, $c\!=\!1$, and $d\!=\!0$.
\end{proof}

% \section*{Acknowledgement}
%  We thank Andrea de Jesus Torres for the fruitful discussion about the detection analysis. 

%============================
% References
%============================
\bibliographystyle{IEEEtran}
% Generated by IEEEtran.bst, version: 1.14 (2015/08/26)

%\bibliography{main.bib}
%\nobibliography*

\end{document}